\documentclass[preprint,12pt]{elsarticle}
\usepackage{amsmath,amssymb,amsthm,mathtools}
\usepackage{subcaption} 
\usepackage{float}
\usepackage{booktabs}
\usepackage{siunitx}
\usepackage{microtype}    
\usepackage[
pdftex,
pdftitle={An analytical approach to calculating stationary PDFs for reflected random walks with an application to BESS-based ramp-rate control},
pdfauthor={Carlos Colchero, Diego Jiménez-Arreguín, Álvaro Herrera, Jorge E. Pérez-García, Oliver Probst},
colorlinks=true,
linkcolor=blue,
citecolor=blue,
urlcolor=black,
]{hyperref}
\usepackage[nameinlink,noabbrev]{cleveref} 
\usepackage{adjustbox}    
\usepackage{url}          
\usepackage{makecell}
\usepackage{parskip}
\theoremstyle{plain}      
\newtheorem{theorem}{Theorem}[section]       
\newtheorem{assumption}[theorem]{Assumption} 

\newtheorem{lemma}[theorem]{Lemma}

\theoremstyle{definition} 

\theoremstyle{remark}  
\newtheorem{remark}[theorem]{Remark}

\journal{Applied Numerical Mathematics}

\begin{document}
\begin{frontmatter}

\title{An analytical approach to calculating stationary PDFs for reflected random walks with an application to BESS-based ramp-rate control}

\author[inst1]{Carlos Colchero}\fnref{fn1}
\author[inst1]{Diego Jiménez-Arreguín}
\author[inst1]{Álvaro Herrera}
\author[inst1]{\\ Jorge E. Pérez-García\corref{cor1}}
\author[inst1]{Oliver Probst\corref{cor1}}
\cortext[cor1]{Corresponding author.}
\ead{oprobst@tec.mx}

\affiliation[inst1]{organization={School of Engineering and Sciences, Tecnológico de Monterrey},%
  city={\\ Monterrey}, state={Nuevo León}, postcode={64700}, country={México}}

\begin{abstract}
A Wiener–Hopf–type integral equation for the stationary PDF of a reflected random walk is derived rigorously based on modern probability theory, and an application to battery energy storage systems (BESS), specifically the sizing of the inverter, is discussed in depth. The methodological steps include the construction of a Markov kernel, the derivation of a Fredholm integral equation of the second kind for the BESS power's PDF, and an analytical solution of the equation based on a Neumann series. The analytical results were compared against numerical solutions obtained with the Nyström method, as well as against the results of an algorithmic simulation using simulated input time series. The use of truncated versions of the analytic solution allows for the construction of simplified design rules for the power systems practitioner. General insights into inverter sizing criteria of storage systems for ramp-rate control of Variable Renewable Energy (VRE) sources such as wind and solar are provided.  
\end{abstract}


\end{frontmatter}

\section{Introduction}\label{sec:intro}

Wind and solar photovoltaic (PV) power plants, also known as variable renewable energy (VRE) sources, are now by far the cheapest electricity generation technologies worldwide \cite{lazard_lcoeplus_2025,owid_lcoe_2025}. In 2024, 92.5\% of all new power plants were renewable \cite{IRENA2025}. The specific challenges posed by the integration of wind and solar into electricity grids can roughly be divided into two categories: (1) challenges related to the fact that wind and solar plants are based on inverters using power electronics, i.e., they are inverter-based resources (IBR), and (2) challenges related to the variable nature of the wind and solar resource. Though variability in itself is not a new phenomenon in electrical power systems, given that the demand varies on a number of different time scales, it typically increases when penetration levels of wind and solar become high. Since a full decarbonization of the power system will evidently require very high levels of VRE technologies, a proper approach to handling this additional variability is required. 

VRE variability manifests itself at different time scales, relevant to different power systems functions. In countries with already high portions of VRE, such as Germany, where the 2024 fraction of VRE was 43\% and the total renewable energy fraction was 57\% \cite{destatis_bruttostrom_2024}, one of the main concerns is often \emph{resource adequacy}, i.e., the capability of the system of handling the load in times of very low simultaneous production levels of wind and solar ("dunkelflaute") \cite{BECKBANG2026114910}. In other systems, such as in California, the main challenge is to meet the daily evening up-ramp requirements for non-solar resources \cite{HOBBS2022100024}, in response to the decaying production from solar PV. In most other systems, however, where VRE penetration is, say, some 15\% or less, the main challenge is (secondary) frequency regulation \cite{ZHANG2025126700}. Grid code requirements related to tolerable frequency excursions are then often translated into maximum tolerable VRE ramps, both upward and downward. Upward ramps can be handled directly by the VRE plant controller by an appropriate curtailment strategy. Downward ramps, on the other hand, require an additional source of power to make up for an inadmissible drop in power from one time step to the other, possibly extending over several time steps. Typically, a battery energy storage system (BESS) is proposed to provide such a ramp-rate control.

\medskip

Ramp-rate control by BES systems has been studied by a number of authors. Marcos et al. \cite{MARCOS201428} proposed their \emph{worst fluctuation} model for a  solar PV power plant, based on the empirical observation that the largest negative ramp was given by an exponential decay from 90\% to 10\% with a time constant $\tau$ given by $\tau = al+b$, with $l$ being a characteristic length of the solar PV array and $a$ are $b$ site-specific fit parameters. The required BESS capacity could then be determined easily for any given tolerable ramp rate. Sukamar et al. \cite{SUKUMAR2018218} reviewed a number of ramp-rate controlling strategies for solar PV plants, classifying the available methods into (a) smoothing, (b) filter-based, and (c) direct ramp-rate control methods. They concluded that methods based on smoothing techniques such as moving-average weighting as well as filtering techniques tend to cause significant battery activity even during periods of low variation. The authors conclude that this effect leads to increased degradation but do not discuss the implications on BESS sizing. Direct ramp-rate control techniques were found to significantly reduce over-utilization of the BESS and, consequently, degradation. Liu et al. \cite{LIU2025117421} implemented a direct ramp-rate control method for solar PV plants, focusing on the alleviation of battery degradation. The authors demonstrated the effectiveness of their methods. However, sizing considerations and possible trade-offs between ramp-rate compliance and aging were not addressed. Montalá et al. \cite{MONTALAPALAU2025114631} also studied the degradation of BES systems designed for ramp rate control and conducted a techno-economic assessment. Most of the discussion was based on sample time series of a BESS working under different market conditions. The authors identify some sizing considerations in the conclusions section but otherwise provide no design considerations beyond the observations of their specific cases.

Though valuable insights have been gained by the works cited, one element is notoriously absent from the relevant scientific literature: a consistent and sufficiently general methodology for the sizing of BESS systems. Most published studies have a mainly algorithmic perspective and use VRE data for illustration purposes, rather than studying the problem of BESS sizing in a systematic manner. One of the few exceptions is the work by Fregosi et al. \cite{Fregosi2023AnAO}, who studied the ramp-rate compliance of a $P_{\mathrm{max}}=100$ MW solar PV power plant as a function of the tolerable ramp rate and the power and energy capacities of the BESS using an algorithmic approach. Ramp rate limits of 2.5\% to 10\% $P_{\mathrm{max}}$ per 10-min interval were considered in the simulations. One important finding was that in their base case (10\% $P_{\mathrm{max}}$ tolerable power output change in 10-min intervals), a BESS power capacity $b_0$ of about 20\% of $P_0$ and an energy capacity of $0.1 \mathrm{h}\times P_{\mathrm{max}} = 0.1 \mathrm{h}\times b_0/20\% = 0.5 \mathrm{h}\times b_0$ were enough to mitigate over 97\% of all ramp-rate violations. If forecasting, assumed to be ideal, were considered in the BESS control, the compliance for this BESS setup increased to about 99\%. While these results and the systematic approach by the authors provided interesting insights, it is difficult to deduce general sizing recommendations from them, given the case study nature of the work; moreover, the input data were not revealed. A different approach was taken by Probst \cite{Probst_2025}, who used classical probability theory to calculate lower limits for the required power $b_0$ and energy capacity $E_0$ of storage systems for ramp-rate control. Analytical results for $b_0$ and $E_0$ were obtained under the assumption of the VRE output variations following both simple and generalized Laplace distributions. The work assumed the BESS power to depend only on the fluctuations of the primary (VRE) power source, i.e., it neglected correlations between consecutive BESS power values. While this is asymptotically exact for large ramp-rate limits, it underestimates the required BESS power capacity at smaller, i.e., stricter, ramp-rate limits.

\medskip

In the present work, we extend the systematic analytical approach proposed by Probst \cite{Probst_2025} by explicitly considering correlations between consecutive BESS power values. We will focus on the sizing of the BESS inverter, i.e., the \emph{power} capacity of the system. A solution to the energy capacity sizing problem will be presented in a companion paper. The basic idea is to derive a probability density function (PDF) for the battery power based on the knowledge of the PDF of the VRE fluctuations. Then, the required BESS inverter capacity can be calculated as a suitable $p$-value of the battery power PDF. The ramp control model considered ensures strict compliance with a given yet arbitrary, ramp rate limit. As will be shown below, this leads to a BESS power process that can be characterized as a reflected random walk possessing the Markov property, i.e., the BESS power at a time step $t_n$ only depends on its value at the previous time step $t_{n-1}$.

\medskip

While reflected random walks have not been used previously in the modeling of BESS processes, they are quite common in the field of operations research, particularly in the theory of queues \cite{kleinrock1975queueing}. An early and still influential work on the theory of queues with a single server was published by Lindley in 1952 \cite{Lindley_1952}. The author derived an integral equation for the distribution of the waiting times of customers in a single queue based on the distribution of a random variable given by the difference between the service time and the inter-arrival time between customers. A number of extensions to and applications of Lindley's work have been published, including the extensions to higher-order processes by Karpelevich et al. \cite{KARPELEVICH199465}, the use of selective function for the modeling of the kernel of the Lindley equation by Kartachevski \cite{Kartashevskiy_2018}, and an extension to branching random walks by Biggins \cite{BIGGINS1998105}. An analytical solution of the original Lindley equation for the case of polynomial distributions and an approximation for the case of finite support were derived by Vlasiou and Adan \cite{VLASIOU2007105}. Very recently, Lucrezia et al. \cite{Lucrezia_Lindley_2025} proposed a solution of the Lindley equation for Laplace jumps in terms of an infinite expansion and applied their formula in the context of the CUSUM method for the detection of changes in a distribution. 

\medskip

In the present manuscript, we generalize Lindley's work to arbitrary Markov processes using modern probability theory \cite{Durrett2019}, specifically the underlying concepts of measure theory, apply the results to BESS power processes for ramp control and obtain a rigorous analytical solution for the BESS power probability density function (PDF) in the case of negative ramp control. The solution is obtained in terms of a Neumann series, allowing for a very compact formulation of the solution. The general concepts, the nomenclature, and a representative BESS process for ramp control with a strict enforcement of the tolerable slope are presented in section \ref{subsec:general_model}. The required fundamentals of measure theory are briefly reviewed in section \ref{subsec:measure_theory}. An update law or recursion in discrete time for the BESS power cumulative density function (CDF) is derived, which is then used as a starting point for the derivation of both the absolutely continuous and the "atomic" or point mass part of the probability density function (PDF) (subsection \ref{subsec:PDF}). The Radon-Nikodym theorem is used to obtain update laws for the absolutely continuous and atomic parts in the form of coupled Fredholm integral equations of the second kind. While these equations have general validity for Markov processes, the rest of the manuscript then focuses on a BESS process for negative ramp control with strict compliance; this process is described in subsection \ref{subsec:Model_negative_ramps}. Update laws for the time-varying PDFs for this case are obtained in subsection \ref{subsec:CDF_PDF_negative_ramp_control} and are then used to derive integral equations of the Wiener-Hopf type for the stationary PDF and CDF of the BESS power in subsection \ref{subsec:stationary_law}. A rigorous analytical solution of the Wiener-Hopf equation for the BESS power PDF in terms of Neumann series for the case of Laplace-distributed primary power fluctuations is described in subsection \ref{subsec:analytical_solution}; the full details of the derivation are provided in \ref{app:coeff}. The numerical solution of the Wiener-Hopf-type equation using the Nyström method, used for validation, is not considered an original contribution of this work and is therefore relegated to  \ref{app: Nystrom}. A brief discussion of the setup of the algorithmic solution used for validation of the analytical results obtained in this work is provided in subsection \ref{subsec:algorithmic_solution}. The results presented in section \ref{sec:results_discussion} include a comparison of the analytical results with those obtained from the numerical solution of the Wiener-Hopf equation and the results of an algorithmic implementation of the problem. The run times of the analytical and the numerical solutions are compared. The use of a truncated version of the analytical solution for a first estimation of the BESS power capacity is demonstrated. A summary with conclusions, including an outlook to future work, is presented in section \ref{sec:summary_conclusions}.

To the best knowledge of the authors, no mathematically rigorous solution to the BESS sizing problem considering correlations between consecutive BESS power values has been published so far, and very few attempts to systematically assess the BESS requirements for ramp control have been made in literature. We therefore believe that this work fills a significant void in the scientific literature on BESS sizing and design.

\section{Methodology}

\subsection{General BESS model for ramp control}\label{subsec:general_model}
The main objective of this work is the probabilistic description of the battery storage system's response to changes in VRE output power, given an arbitrary but fixed ramp-rate limit. Let $(P_n)_{n\ge
0}$ [MW] be a sequence of discrete VRE output power measurements. The time step between consecutive measurements, $\Delta t = t_{n+1}-t_n$, will be assumed to be constant. The results obtained in this work are independent of the value of $\Delta t$; however, for practical applications, it may be useful to think of $\Delta t$ of being of the order of a few seconds to about one minute for primary frequency regulation (PFR) and in the range between 1 and a few minutes for secondary frequency regulation (SFR). A typical time step is $\Delta t = 1$ min. The changes in VRE output power [MW], also referred to as the primary power changes, can be written as:
\begin{align}
    Y_n = P_n-P_{n-1}.
    \label{DEF:CHANGES}
\end{align}

Ramp-rate limits $\alpha_{\uparrow}>0$ and $\alpha_{\downarrow}>0$, both measured in MW/min, may be established for individual power plants for up- and downramps, respectively, ideally based on a rigorous analysis of frequency variations and the system bias \cite{Ayaz2025}. A time-step independent formulation can be obtained by introducing the maximum tolerable step change $a_{\uparrow,\downarrow}=\alpha_{\uparrow,\downarrow} \Delta t$ [MW] in grid power, i.e., the power injected into the grid by the power plant, at any given time. The grid power $R_n$ is defined as:
\begin{equation}
    R_n = P_n+B_n,
    \label{DEF:GRID}
\end{equation}
where $B_n$ is the power of an energy storage system, typically a BESS, at time step $t_n$. BESS action is required if the change between the current power value $P_n$ and the grid power value at the previous interval $R_{n-1}$ is more negative than the tolerable slope $-a_\downarrow$ or greater than the positive step change $a_\uparrow$. It is easy to show that the required battery power is then given by
\begin{equation}
    B_{n+1}=\begin{cases}
    B_n + (a_\uparrow - Y_n), &\Delta_n > a_\uparrow\\ 
    0, & -a_\downarrow \le \Delta_n \le a_\uparrow\\
    B_n - (a_\downarrow + Y_n), & \Delta_n < -a_\downarrow,
    \end{cases}
    \label{eq:BESS_recursion}
\end{equation}
where $\Delta_n = P_n-R_{n-1}$. In this work, we will assume that the power changes $Y_n$ are independent and identically distributed (IID), i.e., a possible autocorrelation between the primary power increments will be neglected. This is generally a good approximation for time intervals of 1 min or shorter \cite{Probst_2025}.  In such cases, the evolution of the BESS power can be described as a random walk driven by stochastic increments $Y_n$. In the following, we will show that this information is enough to relate the probability density function (PDF) of the battery power with the PDF of the primary power fluctuations $Y_n$. In order to do so, a brief review of some basic concepts of modern probability will be provided first.

\subsection{Elements of modern probability theory}\label{subsec:measure_theory}

From the inspection of Eq.~\eqref{eq:BESS_recursion}, it can be seen that the BESS process can be characterized as a Markov process. We work on a probability space $(\Omega, \mathcal{F},\mathbb{P})$ capable of supporting an infinite sequence of random variables $(Y_n)$ with this CDF. Here, $\Omega$ is the so-called sample or outcome space, $\mathcal{F}$ is the event space, with an event being a subset of $\Omega$, and $\mathbb{P}$ a probability function assigning a value between 0 and 1 to each event. 

The increment law $F_Y$ induces a probability measure $\nu$ on $(\mathbb{R},\mathcal{B}(\mathbb{R}))$, defined by
\begin{equation}
    \nu((-\infty, y])=F_Y(y),
\end{equation}
where $\mathcal{B}(\mathbb{R})$ is the Borel sigma algebra on $\mathbb{R}$. Since the increment sequence is IID, the joint law of the entire sequence is given by the product measure $\mathbb{P}=\nu^{\otimes \mathbb N}$ so that each coordinate map $Y_n(\omega)=\omega_n$ is an IID random variable with law $\nu$ for $\omega\in\Omega$.

The battery state sequence $(B_n)_{n\ge0}$ is a family of random variables defined recursively from the increments by a deterministic dispatch rule
\begin{equation}
    B_{n+1}=\Phi (B_n, Y_{n+1}), \quad B_0=b_0,
\end{equation}
where $\Phi: S \times \mathbb R \rightarrow  S$ represents the battery's update mechanism. Equation (\ref{eq:BESS_recursion}) is an example of $\Phi$. $S$ is a fixed measurable state space, for example, $S=[0,\infty)$ in models providing negative ramp control without an upper battery limit or $S=[0, B_\mathrm{max}]$ for finite-capacity batteries. Therefore, we assume that the state space does not evolve with time and that the process remains well-defined:
\begin{equation}
    \Phi(s,y)\in S \quad \text{for all } (s,y)\in S\times \mathbb R.
    \label{HOMO_ASSUM}
\end{equation}

Let $\mathcal{F}_n=\sigma(Y_1, \cdots, Y_n)\subset \mathcal F$ denote the sigma algebra generated by the first $n$ power increments. Intuitively, $\mathcal F_n$ represents all information available up to time $n$. Because $B_n$ is computed from $(Y_n)_{n\ge1}$, it is a $\mathcal{F}_n$-measurable random variable in the same probability space $(\Omega, \mathcal F, \mathbb P)$. Every random variable induces a pushforward measure (law) on its state space so that the probability for the battery power to be in a certain outcome set $A$ ("battery law") is given by
\begin{align}
    \Pi_n(A):=\mathbb P(B_n\in A)=\mathbb P(B_n^{-1}(A)), \quad A\in\mathcal B(S).
\end{align}
where $\mathcal B(S)$ denotes the Borel sigma-algebra on $S$. When $S\subseteq \mathbb R$, we may describe the law through a cumulative distribution function:
\begin{equation}
    G_n(b)=\Pi_n((-\infty,b])=\mathbb P \{ B_n \le b\}, \quad b\in S.
\end{equation}

The pair $(\Pi_n, G_n)$ characterizes the statistical state of the battery at time $n$. Since $B_{n+1}$ depends only on $B_n$ and the new increment $Y_{n+1}$, and since $Y_{n+1}$ is independent of the past, $\mathcal F_n$, the sequence $(B_n)_{n\ge0}$ satisfies the Markov property:
\begin{align}
    \mathbb{P}(B_{n+1}\in A | \mathcal{F}_n)=\mathbb{P}(B_{n+1}\in A |B_n), \quad A\in \mathcal{B}(S).
\end{align}

Hence, recalling the assumption in Eq.~\eqref{HOMO_ASSUM}, $(B_n)$ is a time-homogeneous Markov chain on the state space $S$. Therefore, there exists a Markov transition kernel
\begin{equation}
    \kappa:S\times \mathcal B (S) \rightarrow [0,1], \quad \kappa(s,A)=\mathbb P\{\Phi(s,Y)\in A\}.
    \label{KERNEL}
\end{equation}

Under regular conditional probability, we may express the Markov kernel as:
\begin{equation}
    \mathbb P(B_{n+1}\in A \mid \sigma(B_n))=\kappa(B_n,A)\quad \text{a.s.}
\end{equation}
so that for each fixed $A\in \mathcal B(S)$, the map $s\rightarrow \kappa(s,A)$ is Borel-measurable; for each fixed $s\in S$, the map $A\rightarrow \kappa(s,A)$ is a probability measure. 

The Markov kernel defines an integral operator on finite measures, so that for $A\in\mathcal B(S)$,
\begin{equation}
    (\mu \kappa)(A):= \int_S \kappa(s,A)\ \mu(\mathrm{d}s).
    \label{KERNEL:UPDATE_MEASURES}
\end{equation}

\begin{lemma}
    Let $(S, \mathcal{B}(S))$ be a measurable state space and $\kappa:S\times \mathcal{B}(S)\rightarrow [0,1]$ a Markov transition kernel. Then, for every Borel set $A\subseteq S$,
    \begin{align}
        \Pi_{n+1}(A)=(\Pi_n\kappa)(A)=\int_S \kappa(s,A)\ \Pi_n(\mathrm{d}s).
        \label{eq:BESS_law}
    \end{align}

    When the state space $S\subseteq \mathbb R$, we may describe this update law in terms of the cumulative distribution function $G_n(b)=\Pi_n((-\infty, b])$. Defining $F_\kappa (b|s):=\kappa(s,(-\infty, b])$, the CDF satisfies
    \begin{equation}
        G_{n+1}(b)=\int_S F_\kappa (b|s)\ \Pi_n(\mathrm{d}s).
        \label{eq:BESS_law_2}
    \end{equation}
    \label{LEMMA UPDATE}
\end{lemma}

The proof of this lemma is shown in \ref{app:proofs}. While an expression for the CDF provides a general treatment for any dispatch recursion, it is often the case that the CDF expression turns into an integrodifferential equation. For that reason, we shift our analysis to the probability density function for the cases when it exists.

\subsection{Probability Density Function} \label{subsec:PDF}

Let $(S,\mathcal B(S))$ denote the fixed state space of the battery process, where, $S\subseteq \mathbb R$ unless otherwise specified. We assume that all relevant measures, such as the given CDF-induced $\nu$, have no singular continuous component. Therefore, each measure is expressible as the sum of an atomic part (in the sense of measure theory) and an absolutely continuous part with respect to the Lebesgue measure. Under this assumption, we may construct a single dominating reference measure and define Radon-Nikodym (RN) derivatives for measures on the given state space.

Let $C\subset S$ denote a countable set of atomic locations (i.e., locations with a point measure), all constant in time. These sorts of atomic portions are generated by bounds on the state space, such as $C=\{0\}$ in the negative ramp control model with an upper battery limit or $C=\{0,B_{max}\}$ in a finite capacity model. We define the dominating measure
\begin{equation}
    \rho:=\lambda _{S \backslash C}+\sum_{c\in C}\delta _c,
\end{equation}
where $\lambda$ denotes the Lebesgue measure restricted to $S \backslash C$, i.e., the continuous part of $S$, and $\delta_c$ denotes the Dirac measure on the atom, $c$. 

\begin{remark}
    Every measure in the family $\{\kappa(s,\cdot), \Pi_n(\cdot)\}$, i.e., the Markov transition kernel and the probability of encountering the BESS power in a certain subset of the outcome space, can be split into an absolutely continuous and an atomic portion.These measures are dominated by $\rho$, i.e. $\Pi_n\ll \rho$ and $\kappa(s,\cdot)\ll \rho$ for every $s\in S$.
    \label{REMARKSPLIT}
\end{remark}

Since $\rho\gg \Pi_n$, by the Lebesgue Decomposition Theorem, each battery law (i.e., the probability of encountering the BESS power in the outcome subset $A$) can be written as the following decomposition into an atomic part (again, in terms of measure theory) and an absolutely continuous part:
\begin{align}
    \Pi_n(A)=\sum_{c\in C} p_{c,n}\delta_c(A) + \int_{A\cap (S\backslash C)}g_n(s)\ \mathrm{d}s, \quad A\in \mathcal B (S),
\end{align}
where we define
\begin{equation}
    p_{c,n}:=\Pi_n(\{c\}), \quad g_n(s):=\frac{\mathrm{d}(\Pi_n|_{S\setminus C})}{\mathrm{d}\lambda}(s), \quad s\in S\setminus C.
\end{equation}

Since $\rho\gg\Pi_n$, the conditions for calculating the RN derivative of the BESS power probability are fulfilled and the latter can be obtained as
\begin{equation}
    \frac{\mathrm{d}\Pi_n}{\mathrm{d}\rho}(s)=\sum_{c\in C}p_{c,n}\textbf{1}_{\{c\}}(s)+g_n(s)\textbf{1}_{S\backslash C}(s),
    \label{RNLAW}
\end{equation}
where $\textbf{1}_{\{c\}}$ and $\textbf{1}_{S\backslash C}(s)$ are indicator functions, defined by $\textbf{1}_{A}(x)=1$ if $x\in A$ and $\textbf{1}_{A}(x)=0$ if $x\notin A$, for any arbitrary subset $A$ of the outcome space. Since also $\rho\gg\kappa$, the Markov kernel can also be decomposed into an atomic and an absolutely continuous part:
\begin{equation}
    \kappa(s,A)=\sum_{c\in C}a(c|s)\delta_c(A)+\int_{A\cap (S\backslash C)}k(b|s)\ \mathrm{d}b,
    \label{eq:Kernel_decomposition}
\end{equation}
where we define
\begin{equation}
    a(c|s):=\kappa(s,\{c\}), \quad k(b|s):=\frac{\mathrm{d}\kappa(s, \cdot)}{\mathrm{d}\lambda}(b) \text{ on } S\backslash C.
    \label{eq:Kernel_decomposition_defs}
\end{equation}

The Radon-Nikodym derivative with respect to $\rho$ can then be calculated as:
\begin{align}
    \frac{\mathrm{d}\kappa(s,\cdot)}{\mathrm{d}\rho}(b)=\sum_{c\in C}a(c|s)\textbf{1}_{\{c\}} (b) + k(b|s)\textbf{1}_{S\backslash C}(b).
    \label{RNKERNEL}
\end{align}

After these preparations we are now in the position to state expressions for the probability density function (PDF) that explicitly account for the split between the atomic and the absolutely continuous parts. We first observe that the expression for the CDF of the BESS power (Lemma \ref{LEMMA UPDATE}, with equations (\ref{eq:BESS_law}) and (\ref{eq:BESS_law_2})) can be derived in a Radon-Nikodym sense to obtain:
\begin{equation}
    \frac{\mathrm{d}\Pi_{n+1}}{\mathrm{d}\rho}(b)
= \int_S \frac{\mathrm{d}\kappa(s,\cdot)}{\mathrm{d}\rho}(b)\ \Pi_n(\mathrm{d}s)
= \int_S \frac{\mathrm{d}\kappa(s,\cdot)}{\mathrm{d}\rho}(b)\ \frac{\mathrm{d}\Pi_n}{\mathrm{d}\rho}(s)\ \rho(\mathrm{d}s).
\label{eq:RN_update}
\end{equation}
where $\mathrm{d}\kappa(s,\cdot)/\mathrm{d}\rho$ is the RN derivative of the Markov transition kernel. Expressions for the evolution of both the atomic and the absolutely continuous parts of the PDF can be obtained by inserting Eqs.~\eqref{eq:Kernel_decomposition}, (\ref{eq:Kernel_decomposition_defs}) and (\ref{RNKERNEL}) into Eq.~(\ref{eq:RN_update}). Observing the indicator functions, the following update laws are obtained:
\begin{align}
    p_{c,n+1}=&\sum_{c'\in C}p_{c',n}a(c|c') + \int_{S\backslash C} a(c|s) g_n(s)\ \mathrm{d}s, \quad c\in C, \\
    g_{n+1}(b)=&\sum_{c'\in C}p_{c',n}k(b|c')+\int_{S\backslash C} k(b|s)g_n(s)\ \mathrm{d}s, \quad b\in S\backslash C.
    \label{mastereq}
\end{align}

These equations remain valid for any recursion $B_{n+1}=\Phi(B_n, Y_{n+1})$ as long as the dynamics preserve a fixed countable atomic set $C$.

\subsection{Model for controlling negative ramps} \label{subsec:Model_negative_ramps}

In the rest of this paper, we will focus on the case of negative ramp control. This case is usually one of the most practical interest, since it directly attends to the system operator's interest in managing (relatively) sudden decreases in primary VRE power for the sake of frequency control. It will be argued that no fixed upper limit for the battery power has to be considered; rather, the required battery power will emerge naturally as a suitable $p$-value (say, P99) of the stationary BESS power distribution.
In this case, the more general recursion of equation (\ref{eq:BESS_recursion}) reduces to:
\begin{equation}
    B_{n+1}=\max(-\Delta_n-a,0).
    \label{DEF:BESS}
\end{equation}

It is easy to see that the change in grid power $Z_n=R_n-R_{n-1}$ will remain in compliance with the tolerable slope (i.e., $Z_n\ge -a$) if the BESS complies with equation \ref{DEF:BESS}. It should be noted that in the equations above we have assumed that the corrective action provided by the BESS is instantaneous, i.e., the response time of the controller can be neglected compared to the duration of the time interval $\Delta t$. For a typical time scale of $\Delta t=1$\ min this is generally a good approximation. The effect of a finite response time of the BESS controller will be assessed in future work. 

We may rewrite Eq.~\eqref{DEF:BESS} by noting that $\Delta_n=Y_n-B_{n-1}$, yielding
\begin{equation}
    B_{n+1}=\max(B_n-Y_{n+1}-a, 0).
\end{equation}

We assume $(Y_n)_{n\ge 1}$ is IID, and that the associated CDF is absolutely continuous. We define the effective increments as $X_n:=-Y_n-a$, simplifying the recursion to the well-known reflected random walk (also known as a Lindley recursion \cite{Lindley_1952}):
\begin{equation}
    B_{n+1}=\Phi(B_n, X_{n+1})=\max(0,B_n+X_{n+1}).
\end{equation}

The mean of $X_n$, referred to as the drift, is given by $\mathbb{E}[X_N]$. The name suggests that the average behavior of $X_n$ determines the movement of the reflected random walk, which will be fundamental for the derivation of stationary distributions. For that purpose, the following assumptions must hold:
\begin{assumption}
The sequence $(X_n)_{n\ge1}$ of effective VRE power increments is IID with negative drift.
\label{MOD:IID_ASSUM}
\end{assumption}
In the present case, it can be seen that this condition is fulfilled if the plausible assumption of primary power fluctuations with zero mean (i.e., $\mathbb{E}[Y_N]=0$) is met. Then, $\mathbb{E}[X_N] = -a <0$, as required.

\begin{assumption}
The BESS can only discharge or turn off, and it does so instantly.
\label{MOD:DISCHARGE_ASSUM}
\end{assumption}

 This assumption is implicit in equation (\ref{DEF:BESS}), if we further assume that this equation reflects the full dynamics of the system and is therefore stated for completeness only.
 
\begin{assumption}
    The BESS has infinite capacity, allowing it to correct all ramp-rate incompliances given a tolerable slope $\alpha=a/\Delta t$ of arbitrary value.
    \label{MOD:CAPACITY_ASSUM}    
\end{assumption}

\subsection{The CDF and PDF of the negative ramp control model} \label{subsec:CDF_PDF_negative_ramp_control}

In the case of the model for negative ramp-rate control, the Markov kernel is as given by Eq.~(1),ref{KERNEL}
\begin{align}
    \kappa(s,A)=\mathbb P\{\Phi(s,X) \in A\}=\mathbb P\{\max(0,s+X)\in A\},
\end{align}
where the subscript $n$ was removed from $X$ due to the IID assumption. We may find the cumulative probability of encountering the BESS power $b$, conditioned on the previous value $s$, i.e., $F_\kappa$, directly as
\begin{equation}
    F_\kappa(b|s):=\kappa(s,(-\infty, b]) = \mathbb P\{\max(0,s+X)\le b\}.
    \label{eq:conditional_CDF}
\end{equation}

Equation~\eqref{eq:conditional_CDF} can be evaluated in a piecewise way:
\begin{enumerate}
    \item If $b<0$, then $F_\kappa(b|s)=0$, since $\Phi(s,X)=\max(0,s+X)\ge0$.
    \item If $b=0$, we can argue the following way: Since the battery power is always zero or positive, we have $s\ge0$ and $\Phi(s,X)\ge0$, and therefore
    \begin{align}
    \nonumber
        \{\Phi(s,X)\le 0\}=&\{\max(0,s+X) =0\}\\
        \nonumber
        =&\{s+X\le 0\}\\
        \nonumber
        =&\{X\le -s\},
    \end{align}
    so that
    \[F_\kappa(0|s)=\mathbb P \{X \le -s \}=F_X(-s).\]
    \item If $b>0$, then
    \[\{\Phi(s,X)\le b\}=\{\max(0,s+X)\le b\}=\{s+X\le b\},\]
    since $b>0$, so that
    \[F_\kappa(b|s)=\mathbb P \{X \le b-s \}=F_X(b-s).\]
\end{enumerate}

Therefore, the Markov kernel can be written as
\begin{align}
    F_\kappa(b|s)=\begin{cases}
        0, \quad &b<0, \\
        F_X(-s), \quad &b=0,\\
        F_X(b-s), \quad &b>0.
    \end{cases}
\end{align}

The CDF update rule in Lemma \ref{LEMMA UPDATE} is then given by:
\begin{align}
    p_{0,n+1}=G_{n+1}(0)=&\int_{[0, \infty)} F_X(-s)\ \Pi_n(\mathrm{d}s), \\
    G_{n+1}(b)=&\int_{[0,\infty)} F_X(b-s)\ \Pi_n(\mathrm{d}s).
\end{align}

As mentioned at the end of Sec.~\ref{subsec:PDF}, manipulating the measure $\Pi_n(\mathrm{d}s)$ requires further assumptions to ensure the split between the atomic and absolutely continuous components. Therefore, we will turn to the PDF update equations derived above instead. In the case of the negative ramp control model, there is only one atomic component, corresponding to $b=0$, i.e., $C=\{0\}$. The expression for the dominating measure $\rho$ can then be written as
\begin{equation}
    \rho = \delta_0 + \lambda_{(0,\infty)}.
\end{equation}

We demonstrated that this measure dominates both the BESS update law (i.e., $\Pi_n$) and the Markov transition kernel $\kappa$ for countable-fixed atom sets; in particular, $\kappa \ll \rho$. Then, the RN derivative for the kernel is given by Eq.~\eqref{RNKERNEL}:
\begin{align}
    \frac{\mathrm{d}\kappa(s,\cdot)}{\mathrm{d}\rho}(b)=F_X(-s) \textbf{1}_{\{0\}}(b)+f_X(b-s)\textbf{1}_{(0,\infty)}(b),
\end{align}
where we have made the replacements $a(0|s)=F_X(-s)$ and $k(b|s)=f_X(b-s)$, and we recall that $F_X$ is absolutely continuous. Similarly, Eq.~\eqref{RNLAW} for the derivative of the battery law can be written as:
\begin{align}
    \frac{\mathrm{d}\Pi_n}{\mathrm{d}\rho}(s)=p_{0,n}\textbf{1}_{\{0\}}(s) + g_n(s) \textbf{1}_{(0,\infty)}(s),
\end{align}
given that there is only one atomic component in this case. Substituting these expressions into the evolution [Eqs.~\eqref{mastereq}] yields:
\begin{align}
    p_{0,n+1}=&\ p_{0,n}F_X(0)+\int_0^\infty F_X(-s)\ g_n(s)\ \mathrm{d}s, \\
    g_{n+1}(b)=&\ p_{0,n} f_X(b)+\int_0^\infty f_X(b-s)\ g_n(s)\ \mathrm{d}s, \quad b>0.
\end{align}

\subsection{Stationary law and the integral equation for the PDF}\label{subsec:stationary_law}

As stated in Assumption~\ref{MOD:IID_ASSUM}, the “free” walk $(X_n)_{n\ge1}$ has negative drift, $\mathbb{E}[X_1]<0$. Therefore, an unrestricted random walk depending on $X_n$, such as $B_{n+1}=B_n+X_{n+1}$, would drift to $-\infty$. The battery process given by Eq.~\eqref{DEF:BESS}, however, is a \emph{reflected} random walk, keeping the state from going negative, but still featuring a downward pressure towards zero. Each time $B_n$ hits zero, the next excursion restarts with new increments that are independent of the past.

Each excursion, i.e., a period during which the battery is injecting non-zero power into the grid, is independent of each other due to the IID assumption; with negative drift, the returns to zero occur infinitely often with finite duration. As stated by Lindley \cite{Lindley_1952}, whenever there is negative drift for $X_n$, the process $B_{n+1}=\max(0,B_n+X_{n+1})$ admits a stationary distribution for $B_n$. Under stationary conditions, $B_n\sim \Pi$ for all $n$, so the CDF for the BESS power is given by the integral equation:
\begin{equation}
    G(b)=\int_{[0,\infty)} F_X(b-s)\ \Pi(\mathrm{d}s), \ \ b\ge0.
\end{equation}

The integral equations for the CDF and PDF of the BESS power for the stationary case can then be recovered from the evolution equations as:
\begin{align}
            G(b)=&\ p_{0}F_X(b)+\int_0^\infty F_X(b-s)\ g(s)\ \mathrm{d}s, \\
            g(b)=&\ p_0f_X(b)+\int_0^\infty f_X(b-s)\ g(s)\ \mathrm{d}s, \ \ b\in(0,\infty),
            \label{eq:Lindley_CDF_PDF}
\end{align}

where the following normalization condition must be satisfied:
\begin{equation}
    1=p_0+\int_0^\infty g(b)\ \mathrm{d}b.
\end{equation}

The PDF integral equation is known as the Lindley integral equation. 

\subsection{Analytical solution}\label{subsec:analytical_solution}

Having found a general formula for the PDF of the battery discharge process, we may study specific choices for the kernel. If the chosen kernel features exponentially decaying tails, one may define an integral operator by
\begin{equation}
    (\textbf{K}g)(b):=\int_0^\infty f_X(b-s)\ g(s)\ \mathrm{d}s, \quad b\in [0,\infty),
\end{equation}

in an appropriate function space $\mathcal{W}$ such that $||\textbf{K}||_\mathcal{W}<1$. With this definition, the integral equation (\ref{eq:Lindley_CDF_PDF}) can be written as:
\begin{equation}
    g(b)=p_0f_X(b)+(\textbf{K}g)(b) \iff ([\textbf{I}-\textbf{K}]g)(b)=p_0f_X(b),
\end{equation}

where \textbf{I} denotes the identity operator on $\mathcal{W}$ and $g\in \mathcal{W}$. By introducing a re-scaled version $u(b)$ of the BESS power PDF through $g(b)=p_0u(b)$, where $p_0$ is the point mass density as before, the following simple form of the integral equation is obtained:
\begin{equation}
    ([\mathbf{I}-{\mathbf{K}}]u)(b) = f_X( b).
\end{equation} 

Probst \cite{PROBST2020105969} has observed that primary power changes $Y_n$ in wind farms are typically distributed according to a generalized Laplace distribution, which in turn can often be approximated by a simple Laplace (SL) distribution for the purposes of ramp control. Then, assuming $f_Y$ to be SL-distributed, $f_X$ can be written as $f_X(x)=\tfrac{\beta}{2} e^{-\beta|x+a|}$. Introducing normalized quantities by $\tilde{b}=\beta b$, $\tilde{s}=\beta s$ and $\tilde{a}=\beta a$ the following normalized expressions can be obtained:
\begin{align}
    \tilde{f}_X(\tilde{b}-\tilde{s})=&\tfrac{1}{2}e^{-|\tilde{b}-\tilde{s}+\tilde{a}|} \\
    \tilde{f}_X(\tilde{b})=&\tfrac{1}{2}e^{-(\tilde{b}+\tilde{a})}
\end{align}

In normalized variables we then have
\begin{equation}
    (\tilde{\mathbf{K}}g)(\tilde b) = \int_0^\infty \tfrac{1}{2}\, e^{-|\tilde b-\tilde s+\tilde a|}\,g(\tilde s)\,\mathrm{d}\tilde s,
\end{equation}

The operator [$\textbf{I}-\mathbf{\tilde{K}}$] is invertible through a Neumann series, which converges under two conditions: (1) $\mathbf{\tilde{K}}$ acts on the weighted Banach space $\mathcal{C}_\theta=\{g:[0,\infty)\rightarrow\mathbb R: ||g||_\theta < \infty \}$ endowed with the norm $||g||_\theta=\sup_{b\ge 0} \ e^{\theta b}|g(b)|$; and (2) the operator norm satisfies $||\mathbf{\tilde{K}}||<1$. As a consequence of (1), the operator norm is controlled by the momentum-generating function $M_X(\theta)$ of $X$ since $||\mathbf{\tilde{K}}||_{\theta} \leq M_X(\theta)=E[e^{\theta X}]$. (A complete derivation of this estimate is provided in the supplementary repository \cite{BESS_Sizing_Repository}). In particular, the SL case features $M_X(\theta)=e^{-\tilde{a}\theta}/(1-\theta^2)$ for $|\theta|<1$. Due to the negative drift condition, there exists at least one $\theta>0$ that satisfies $M_X(\theta) < 1$ within the neighborhood where $e^{-\tilde{a}\theta}<1-\theta^2$. The solution to the integral equation can then be expressed as
\begin{equation}
    u(\tilde b)=([\mathbf{I}-\tilde{\mathbf{K}}]^{-1}\tilde f_X)(\tilde b) = \sum_{n=0}^\infty \tilde{\mathbf{K}}^{\,n}\tilde f_X(\tilde b).
\end{equation}

In the case of the simple Laplace kernel, the integral operator split at $(\tilde{b}-\tilde{s})+\tilde{a}=0$: 
\begin{equation}
\nonumber
    (\tilde{\mathbf{K}}g)(\tilde b)=\tfrac{1}{2}\!\left[ e^{-(\tilde b+\tilde a)} \!\int_0^{\tilde b+\tilde a} \!e^{\tilde s}\,g(\tilde s)\, \mathrm{d}\tilde s
    + e^{\tilde b+\tilde a}\!\int_{\tilde b+\tilde a}^{\infty}\!e^{-\tilde s}\,g(\tilde s)\,\mathrm{d}\tilde s\right].
    \label{split}
\end{equation}

The action of the operator has therefore been separated into two exponential integrals that can be computed by hand. The partial sums can be readily computed by the repeated application of $\mathbf{\tilde{K}}$. The zeroth term will be denoted as:
\begin{equation}
    u_0(\tilde b)=\tfrac{1}{2}\,e^{-(\tilde b+\tilde a)}.
\end{equation}

Then, the first partial sum is given by:
\begin{equation}
    u_1(\tilde b)=u_0(\tilde b)+(\tilde{\mathbf{K}}\tilde f_X)(\tilde b).
\end{equation}

Calculating $(\tilde{\mathbf{K}}\tilde f_X)(\tilde b)$ implies evaluating the following two terms:
\begin{align}
\nonumber
    \int_{\tilde b+\tilde a}^\infty e^{-\tilde s}\,\tilde f_X(\tilde s)\,\mathrm{d}\tilde s
    =&\int_{\tilde b+\tilde a}^\infty e^{-\tilde s}\,\tfrac{1}{2}e^{-(\tilde s+\tilde a)}\,\mathrm{d}\tilde s
    =\tfrac{1}{4}\,e^{-(2\tilde b+3\tilde a)},\\
        \int_{0}^{\tilde b+\tilde a} e^{\tilde s}\,\tilde f_X(\tilde s)\,\mathrm{d}\tilde s
    =&\int_{0}^{\tilde b+\tilde a} e^{\tilde s}\,\tfrac{1}{2}e^{-(\tilde s+\tilde a)}\,\mathrm{d}\tilde s
    =\tfrac{1}{2}\,(\tilde b+\tilde a)\,e^{-\tilde a}.
\end{align}

With these results, the first partial sum can be written as:
\begin{equation}
\nonumber
    u_1(\tilde b)=u_0(\tilde b)+\tilde{\mathbf{K}}\tilde f_X(\tilde b)
    =\tfrac{1}{2}e^{-(\tilde b+\tilde a)}+e^{-(\tilde b+2\tilde a)}\left(\tfrac{1}{8}+\tfrac{1}{4}(\tilde b+\tilde a)\right).
\end{equation}

Following this pattern, the $n$th partial sum of the infinite series can be represented by:
\begin{equation}
    (\tilde{\mathbf{K}}^{\,n}\tilde f_X)(\tilde b)=e^{-(\tilde b+(n+1)\tilde a)}\sum_{k=0}^n \bar{\Lambda}_{n,k}\,(\tilde b+\tilde a)^k,
\end{equation}

where the $\Lambda_{n,k}$ are real coefficients. The re-scaled BESS power PDF $u(b)=g(b)/p_0$ is then given by the following infinite series:
\begin{equation}
    u(\tilde b)=\sum_{n=0}^{\infty} e^{-(\tilde b+(n+1)\tilde a)}\sum_{k=0}^n \bar{\Lambda}_{n,k}\,(\tilde b+\tilde a)^k = \sum_{n=0}^\infty v_n,
    \label{eq:rescaled_BESS_PDF}
\end{equation}

where we have introduced the variable $v_n$ for later use (equation (\ref{eq:Truncated sums})). In order to find the recursive formula for $\bar{\Lambda}_{n,k}$, we define the unscaled basis functions:
\begin{equation}
    \tilde\phi_{p,k}(\tilde s):=(\tilde s+ \tilde a )^k\,e^{-(\tilde s+p\tilde a)}, \quad p\ge1,\ \ k\ge0,
\end{equation}
The procedure to find the series coefficients can be found in \ref{app:coeff}. Since the kernel splits in two, there are two contributing portions to the coefficient recursion; the "upper" portion was found to be:
\begin{equation}
   U_{k\rightarrow r}=\sum_{m=r}^k \frac{1}{2}\,\frac{k!}{m!}\begin{pmatrix}
        m \\ r
    \end{pmatrix} \tilde a^{\,m-r}\left(\frac{1}{2}\right)^{k-m+1},
\end{equation}

where $r=0,1,\cdots, k$. The lower portion is given by:
\begin{equation}
    L_{k\rightarrow r}=\frac{1}{2}\,\frac{1}{k+1}\begin{pmatrix}
        k+1 \\ r
    \end{pmatrix} \tilde a^{\,k+1-r},
\end{equation}

for $r=1,\cdots, k+1$ and where $L_{k\rightarrow 0}=0$. Finally, the coefficients are given by:
\begin{equation} \label{eq: BESS coef recursion}
    \bar{\Lambda}_{n+1,r}=\sum_{k=0}^n \bar{\Lambda}_{n,k}\,\big(L_{k\rightarrow r} + U_{k\rightarrow r}\big).
\end{equation}

For any $u\in \mathcal{C}_\theta$, the weighted norm bound implies $u\in L^1([0,\infty))$ directly, since $ \mathcal{C}_\theta \subset L^1([0,\infty))$. Consequently, the following holds:
\begin{equation}
    \int_0^\infty |u(b)|\ \mathrm{d}b \le  \frac{1}{\theta}||u||_{\theta}< \infty.
\end{equation}

Therefore, we may perform a probability normalization in the unweighted $L^1$ to find $p_0$:
\begin{equation}
    p_0=\frac{1}{1+\int_0^\infty u(\tilde{b})\mathrm{d}\tilde{b}},
    \label{NormalizationCond}
\end{equation}

so that $g( \tilde{b})=p_0 u( \tilde{b})$. Substituting the expansion of Eq.~\eqref{eq:rescaled_BESS_PDF} into this expression, we obtain:
\begin{equation}
    p_0 = \left(\,1+\int_{0}^{\infty}\sum_{n=0}^{\infty}
        e^{-(\tilde b+(n+1)\tilde a)}
        \sum_{k=0}^{n}\bar{\Lambda}_{n,k}\,(\tilde b+\tilde a)^{k}\,\mathrm{d}\tilde b\right)^{-1}= \frac{1}{1+\tilde\Omega}.
\end{equation}

In \ref{app:point}, we further study this integral expression:
\begin{equation}
    \tilde\Omega \;=\; \int_{0}^{\infty}\sum_{n=0}^{\infty}
        e^{-(\tilde b+(n+1)\tilde a)}
        \sum_{k=0}^{n}\bar{\Lambda}_{n,k}\,(\tilde b+\tilde a)^{k}\,\mathrm{d}\tilde b,
\end{equation}

and arrive at the following series expansion:
\begin{equation}
    \tilde\Omega \;=\; \sum_{n=0}^{\infty} e^{-(n+1)\tilde a}
        \sum_{k=0}^{n}\bar{\Lambda}_{n,k}
        \sum_{j=0}^{k}\binom{k}{j}\,\tilde a^{\,k-j}\, j!=\sum_{n=0}^\infty \tilde{\omega}_n.
        \label{eq:Omega_tilde_full}
\end{equation}

For numerical evaluation of the Neumann series described above, we have generated partial sums up to $M$ terms, i.e.,
\begin{equation}
    u_M=\sum_{n=0}^M v_n\ , \quad \tilde{\Omega}_M = \sum_{n=0}^M \tilde{\omega}_n,
    \label{eq:Truncated sums}
\end{equation}

where $v_n$ and $\tilde{\omega}_n$ were defined in equations (\ref{eq:rescaled_BESS_PDF}) and (\ref{eq:Omega_tilde_full}), respectively. Equations (\ref{eq:Truncated sums}) were evaluated through the application of the coefficient recursion to the polynomial-exponential functions and evaluated on a discrete BESS power value grid. It should be noted that the sums in equations (\ref{eq:rescaled_BESS_PDF}) start at $n=0$, which is why the sums have $M'=M+1$ terms.

\subsection{Numerical solution to the integral equation}\label{sec:numerics}
In order to validate the analytical solution of the Wiener-Hopf integral equations defining the PDF and the CDF of the BESS power, a numerical solution was implemented as well. The complete construction of this solution is explained in \ref{app: Nystrom}. Both the Nyström and the Piccard methods were implemented. For the sake of simplicity, only the results for the Nyström method are reported.

\subsection{Algorithmic solution}\label{subsec:algorithmic_solution}

In order to validate both the analytical and the numerical solutions of the case of negative ramp control studied in this work, a simple algorithmic solution has been implemented as well. Synthetic 1-min wind power series $(P_\mathrm{sim}(t_n))$ were generated by numerically integrating synthetic time series for the power output changes $(Y_n)$ distributed according to the following generalized distribution (see \cite{Probst_2025} and \cite{PROBST2020105969})): 

\begin{equation}
    f_Y(y)=\frac{\gamma}{2}\left(c\ \zeta e^{-\zeta\gamma|Y|}+(1-c)e^{-\gamma|Y|}\right),
    \label{eq:Generalized_Laplace}
\end{equation}

where $c\in [0,1]$ is a free parameter and $\zeta \gg1$. In the present work, $\zeta = 10$ was used as a representative value. The case $c=0$ corresponds to the case of $Y$-values following a simple Laplace (SL) distribution. Versions of equation (\ref{eq:Generalized_Laplace}) with $c>0$ will be referred to as generalized Laplace (GL) distributions. We allowed for such a generalized distribution because realistic wind and solar power changes are often better described by equation (\ref{eq:Generalized_Laplace}) than by simple Laplace distributions. Other possible deviations from ideality, including finite autocorrelation and heteroscedasticity, are beyond the scope of the current work and will be addressed in follow-up work.

In order to be able to compare the results for different values of $c$ the power scale $\beta^{-1}$ was chosen according to \cite{Probst_2025} 

\begin{equation}
    \gamma=\beta\sqrt{\frac{c}{\zeta^2}+(1-c)},
    \label{eq:equivalent_beta}
\end{equation}

where $\beta^{-1}$ is the characteristic power scale of the equivalent SL distribution with the same variance as the GL distribution? All results below will be reported in terms of the normalized variables  $\tilde{a}=a\beta$ and $\tilde{b}=b\beta$, where $a$ continues to be the tolerable change in grid power per time step and $b$ the BESS power, both typically measured in MW.

Synthetic time series for the power changes $(Y_n)$ were generated based on the inverse sampling theorem \cite{devroye2006nonuniform} and then integrated to obtain the power time series $(P_{\mathrm{sim},n})$. Upper ($P_\mathrm{max}$) and lower ($0$) power limits were introduced to account for the finite capacity of the wind farm. Plausible values for $P_\mathrm{max}$ can be obtained by first calculating an equivalent SL power scale $\beta^{-1}=\sigma/\sqrt{2}$ from the standard deviation of the power output variations, which may be available even in the absence of time series with high temporal resolution, and then determining the normalized plant capacity $p_\mathrm{max}$ by $p_\mathrm{max}=P_\mathrm{max}\beta$. In the present work, an empirical value of 0.5 $p_\mathrm{max}=90.9$ was used in all simulations.

\section{Results and discussion}\label{sec:results_discussion}

\subsection{Performance of the analytical solution against the algorithmic reference}\label{sec: analytical vs algorithmic}

The first question naturally arising is how well the analytical solution performs against an algorithmic implementation of the same underlying recursion. The following specific research questions have been addressed: (1) Given plausible values of the normalized critical step change in grid power $\tilde{a}$, how does the analytical solution perform as a function of the number $M$ of terms included in the series expansion? And (2) Does the BESS power PDF indeed depend only on the normalized critical slope $\tilde{a}$, as asserted by the theoretical solution (equation (\ref{eq:Lindley_CDF_PDF})), or do the results depend on $a$ [MW] and the characteristic power scale $\beta^{-1}$ [MW] independently?

In order to address question (1), synthetic wind power time series were generated according to the procedure described in subsection \ref{subsec:algorithmic_solution}. The simulation parameters are summarized in Table~\ref{tab: Sim params}. The tolerable slope \(a\) was defined as \(1\%\) of the rated power \(P_0\) [MW] per time step. In order to get a feel for the severity of the ramp rate limit imposed in this simulation, we note that for the normalized tolerable slope of $\tilde{a}=0.9018$ shown in Table \ref{tab: Sim params} the primary power source violates the ramp rate limit during about 32\% of the time steps; the ramp rate limit therefore has a significant impact on the system.
\begin{table}[h!]
    \centering
    \begin{tabular}{c|c|c|c|c}
      $N$  & $P_0 \; [\mathrm{MW}]$ & $a \; [\mathrm{MW}]$ & $\beta \; [1/\mathrm{MW}]$ & $\tilde{a}$\\
      \hline
      $5\times10^6$ & $150.29$ & $1.503$ & $0.6$  & $0.9018$
    \end{tabular}
    \caption{Parameters used for simulation and analytical evaluation}
    \label{tab: Sim params}
\end{table}

The error between the analytical solution for a specific number $M'=M+1$ of terms of the expansion is measured by taking the $L^1$ norm between the algorithmic PDF and the analytical PDF for the given value of $\tilde{a}$:
\begin{equation}\label{eq: DM norm}
    D_{L^1} = || \tilde{g}(\tilde{b};\tilde{a})^{(M)} - \tilde{g}(\tilde{b};\tilde{a})^{\mathrm{(sim)}} ||_{L^1}
\end{equation}

Figure \ref{fig: PDF and cdf comparision} shows a comparison between the algorithmic and the analytical solution, the latter for a varying number $M \in [0,\infty)$ of the truncated Neumann series ($M=3$, $M=6$ and $M=100$). Both the absolutely continuous part of the PDF and the CDF are shown; note that the CDF includes the point mass of the distribution. It can be seen that  the analytical solution exhibits rapid convergence towards the algorithmic solution. The error between the absolutely continuous part of the PDF, as measured by $D_{L^1}$, is 7.7\%, 2.2\% and 2\% for $M=3$, $M=6$ and $M=100$, respectively. For illustration, the P99-values $\tilde{b}_{99}$ of the normalized BESS power $\tilde{b}$ have been included in the figure. It can be seen that for $M=100$ the agreement between the analytical ($\tilde{b}_{99}=5.15$) and the algorithmic solution ($\tilde{b}_{99}=5.14$) is almost perfect, as long as a large value for the nominal wind farm power is assumed in the algorithmic simulation ($P_\mathrm{max}\rightarrow\infty$). The consideration of a finite wind farm capacity  ($P_\mathrm{max}=150$ MW) produces a slight deviation from the purely Laplacian distribution of the power increments $Y$, leading to a small ($\sim 3\%$) reduction in the P99-value of the battery power. This can be explained by the fact that for output values of the primary power source of $P\sim P_\mathrm{max}$ and $P\sim 0$ the variability is reduced. Using $M=6$ ($M'=7$ terms of the expansion) leads to an underestimation of some $8\%$ of the true P99-value with respect to the reference case ($M=100$, $P_\mathrm{max}\rightarrow\infty)$; using $M=3$ (4 terms of the expansion) increases the error to 17\%.

\begin{figure}[h!]
    \centering
    \includegraphics[width=1\linewidth]{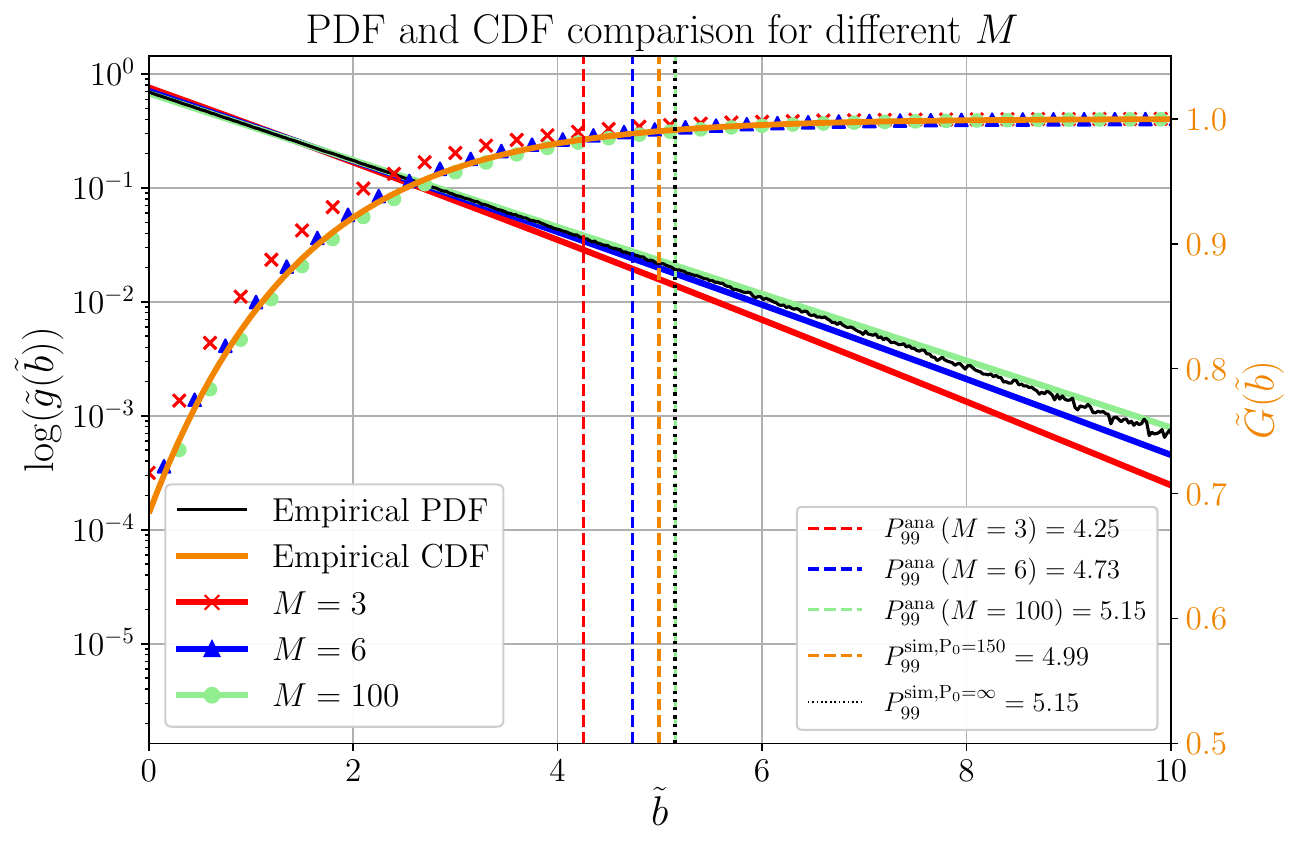}
    \caption{Convergence of the analytical solution with the algorithmic simulation with 3, 6 and 100 $M$ terms of the expansion for the probability density function tail (black) and the cumulative density function (orange)}
    \label{fig: PDF and cdf comparision}
\end{figure}

Further insights into the accuracy of the analytical solution compared to the algorithmic reference can be gained by varying the value of the normalized tolerable slope $\tilde{a}$. As \(\tilde{a}\) decreases the BESS workload becomes larger, leading to an increase of the P99 value of the BESS power, as shown in Figure~\ref{fig: p99 analytical vs Numerical}. As the figure demonstrates, the analytical solution with $M=100$ is an excellent solution to the algorithmic reference for almost the full range of $\tilde{a}$values shown. In the case of $M=100$ shown in Figure~\ref{fig: p99 analytical vs Numerical} (a), a noticeable discrepancy occurs only at very small values of the normalized ramp rate limit \(\tilde{a} \le 0.16\), where the simulation \(P_{99}\) exceeds the analytical estimate by more than 10\%. Evidently, such extremely stringent values have very little practical relevance. For very large values of the tolerable slope (\(\tilde{a} > 3.5\)), on the other hand, the required BESS power capacity is essentially zero.

Regarding the convergence behavior of the analytical solution, it can be noticed (Figure~\ref{fig: p99 analytical vs Numerical} (b)) that the number of required terms for a given error level (taken to be $D_{L^1}=0.05$ in the figure) strongly increases towards small values of $\tilde{a}$. Whereas this compliance level can be met with $M \leq 2$ for $\tilde{a}\ge 1.5$, nearly 200 terms would be needed at $\tilde{a}=0.1$.  

\begin{figure}[h!]
    \centering
    \includegraphics[width=1\linewidth]{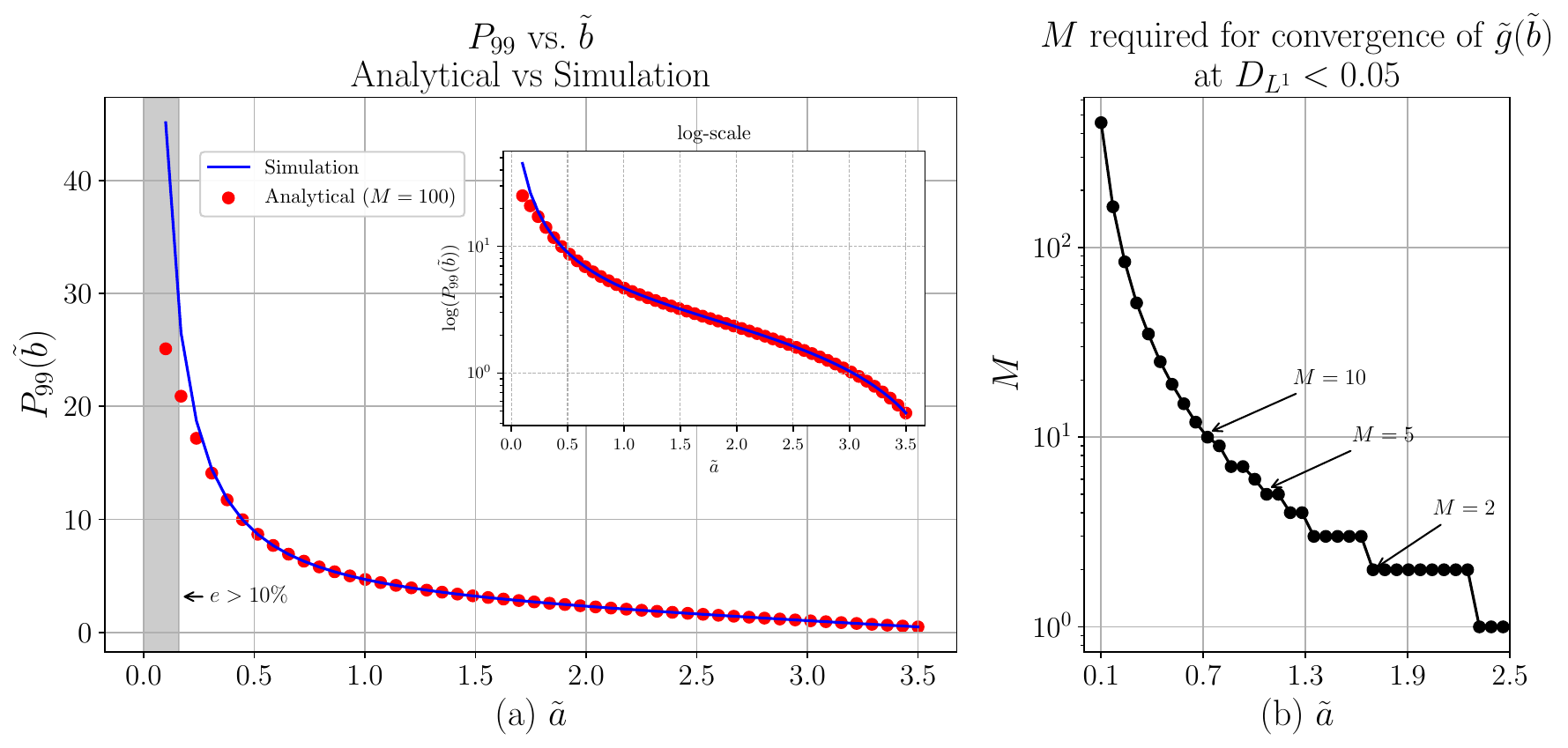}
    \caption{Performance of the analytical solution against the algorithmic battery simulation. a) $P_{99}(\tilde{a})$ comparison between battery simulation and analytical solution with $M=100$. b) Number of terms of the expansion needed to reach convergence with $D_{L^1}\leq0.05$}
    \label{fig: p99 analytical vs Numerical}
\end{figure}

Research question (2), i.e., the question whether the required BESS power capacity only depends on the normalized tolerable slope $\tilde{a}$ rather than on $a$ and $\beta$ independently, was addressed by demonstrating that \(P_{99}(\tilde{b})\) remains invariant for different combinations of \(a\) and \(\beta\) yields the same value of \(\tilde{a}\). Figure~\ref{fig:a_beta_independence} (a) confirms that, as predicted by the theoretical model, the BESS power PDF depends uniquely on the normalized critical slope \(\tilde{a}\), as evidenced by the empirical observation that three different combinations of $a$ and $\beta$ produce the same invariant curve. The combinations selected are illustrated in Figure~\ref{fig:a_beta_independence} (b) where the selected sets of points $(a,\beta)$ have been plotted on top of isolines of constant $\tilde{a}$ in the $a$-$\beta$ plane. Consequently, it can be concluded that the curve \(P_{99}(\tilde{b})\) is unique and fully characterizes the required BESS power capacity for any scenario defined by \(\tilde{a}\).

\begin{figure}[h!]
    \centering
    \makebox[\textwidth][c]{%
        \includegraphics[width=1\textwidth]{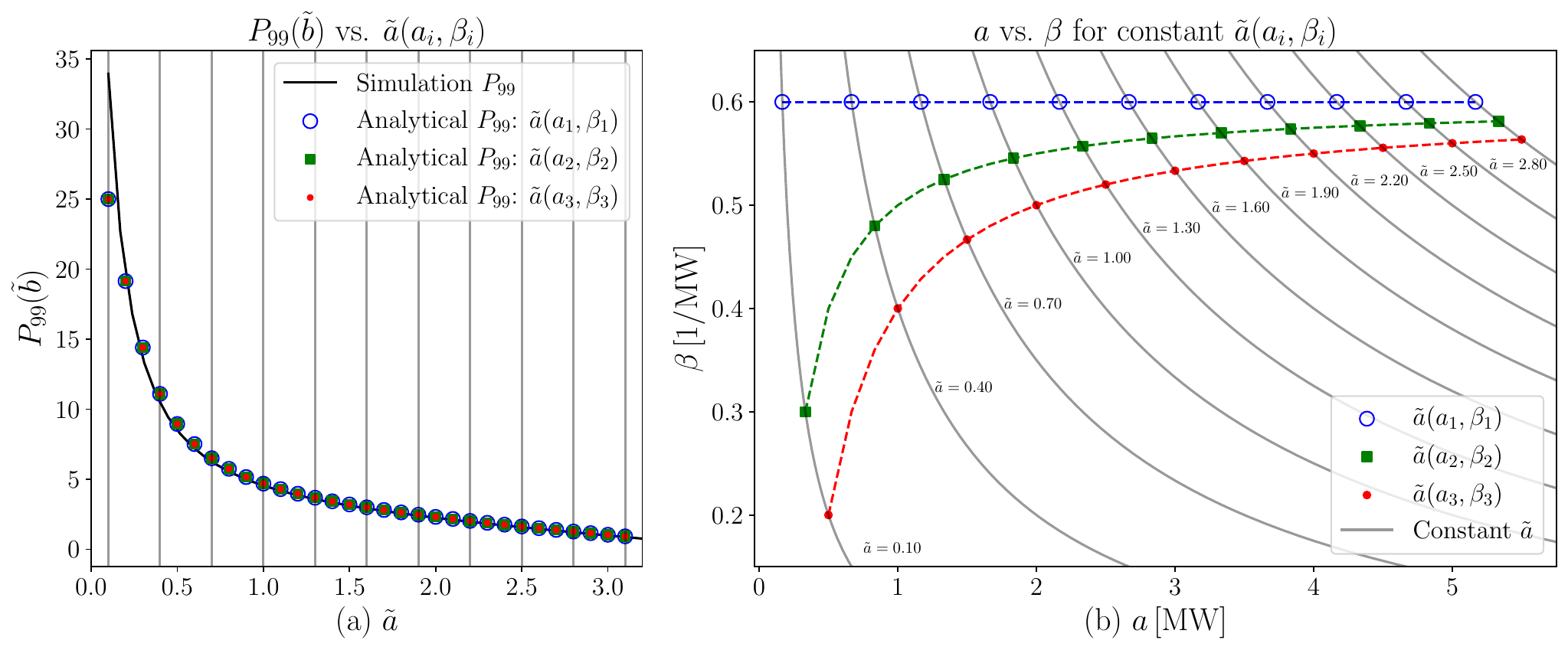}}
    \caption{Validation of $\tilde{g}(\tilde{b},\tilde{a})$ independence on $a$ and $\beta_{SL}$. (a) Comparison of $P_{99}(\tilde{a})$ values for 3 different cases along the $\tilde{a}$ domain. (b) Equivalence of each case following $\tilde{a}=a_1\beta_2=a_2\beta_1=a_3\beta_3$.}
    \label{fig:a_beta_independence}
\end{figure}

\subsection{Comparison of the analytical and the numerical solutions}

While a theoretical solution provides significantly more insight than a numerical solution, it is fair to ask how the accuracies compare. This leads to research question (3): Considering that the analytical solution has been stated in terms of an infinite power series, what are the trade-offs in terms of accuracy vs. computation time? To answer this question, we first observe
(Figure~\ref{fig: RQ2.1-RQ2.2} (a)) that the analytical solution with a modest $M$-value of $M=17$ and the numerical solution with a number of $1000$ grid points used for the evaluation of the Toeplitz operator (\ref{app: Nystrom}) yield practically indistinguishable results between each other for the absolutely continuous part of the BESS power PDF, which are also practically identical to the PDF obtained with the algorithmic solution. For completeness, the P99 values of $\tilde{b}$, which also involve the calculation of the point mass probability $p_0$, have also been indicated in the figure, showing that very similar values are obtained in all three cases.

Regarding the numerical solution, implemented through the Nyström method described in \ref{app: Nystrom}, it is worth noticing that its accuracy is independent of the normalized critical slope $\tilde{a}$ and is instead controlled solely by the grid resolution of the discrete Toeplitz operator. With a number of grid points of $1000$, the Nyström solver achieves a precision of $D_{L^1} \leq 0.01$ for all tested values of $\tilde{a}$, while maintaining computation times below 0.1 s. The analytical method, on the other hand, inherits its computational cost from the structure of the truncated series. As can be seen in Figure \ref{fig: RQ2.1-RQ2.2} (b), the computation time rises steeply with the number of terms $M'$. As we may recall from the inspection of Figure \ref{fig: p99 analytical vs Numerical} (b), the number of terms in the series expansion required for a given level of accuracy strongly depends on $M$, making the computation time dependent on the normalized critical slope $\tilde{a}$. Fortunately, for practical values of $\tilde{a}$, only relatively few terms are required, making the analytical calculation typically more efficient than the numerical one. As can be seen from Figure \ref{fig: p99 analytical vs Numerical} (b), the computation time remains below 0.1 seconds up to $M\sim 75$. Considering the specific case of the simulation parameters listed in Table 1, the Nyström method achieved an accuracy $D_{L^1}=0.0084$ in 0.044 s, while the analytical solver reached a comparable accuracy of $D_{L^1}=0.0094$ using only $M'=18$ terms, with a computation time of 0.004 s on the same computer, i.e., one order of magnitude faster.

\begin{figure}[h!]
    \centering
    \includegraphics[width=1\linewidth]{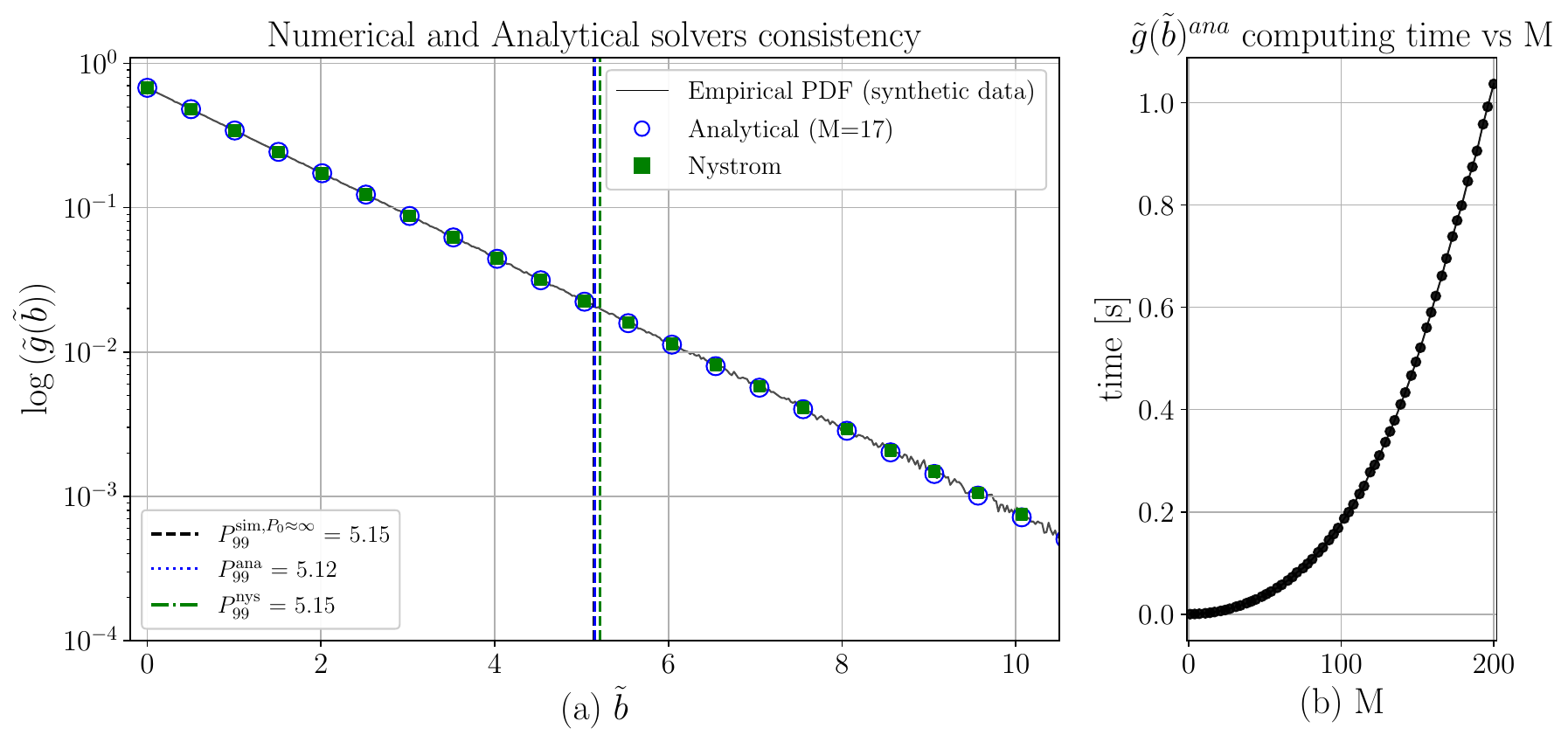}
    \caption{a) Comparison of the algorithmic solution using $P_{sim}^{P_0\approx \infty}$ against Nyström method and Neumann series integral solvers for $\tilde{a}=0.9$. b) Time complexity evolution for the analytical solution in terms of $M$}
    \label{fig: RQ2.1-RQ2.2}
\end{figure}

\subsection{Impact of generalized Laplacian distributions}

As discussed in section \ref{subsec:algorithmic_solution}, real-world wind farm variations are often distributed according to a generalized Laplace distribution of the form given by equation (\ref{eq:Generalized_Laplace}), whereas the theoretical results derived in this paper apply to simple Laplace distributions. A finite value of $c$ in equation (\ref{eq:Generalized_Laplace}) for a given variance of the power increments $Y$ implies a heavier tail of the distribution compared to the $c=0$ case and, consequently, a larger P99-value $\tilde{b}_{99}$. It is therefore interesting to assess how these deviations from ideality influence the required capacity of the BES system.

Fortunately, the impact of heavier tails on the values of $\tilde{b}_{99}$ is relatively modest, as evidenced by Figures \ref{fig: SL vs GL p99} and \ref{fig: RQ3.3}, respectively. Figure \ref{fig: SL vs GL p99} assesses the effect of a finite value of $c$, taken to be $c=0.25$, through the algorithmic solution only. The parameter values in Table~\ref{tab: Sim params} were used once again. As explained in subsection \ref{subsec:algorithmic_solution}, the two simulations were conducted at the same variance level for the power increments $Y$ in order to ensure comparability; this leads to the requirement of $\gamma = 0.537$assuming $\zeta = 10$.

As shown by Figure \ref{fig: SL vs GL p99}, the overall effect of the finite value of $c$ is very small. The P99 values of $\tilde{b}$ are slightly larger for the GL distribution compared to the SL case, as expected from the heavier tail mentioned above, but the difference generally amounts to only a few percent. The relative difference between the two cases increases towards larger values of the tolerable slope $\tilde{a}$, which is readily explained by the fact that for large values of $\tilde{a}$ the BESS, it only operates during rare isolated time steps, leading to the convergence of the absolutely continuous part of the BESS power PDF and the PDF of the shifted negative power increments $X=-Y-a$ \cite{Probst_2025}. Given that the GL-PDF for $X$ has a heavier tail than the SL-PDF, so does the BESS power PDF. At small values of $\tilde{a}$ the dominant effect is the correlation between consecutive values of $B$ due to the Markov process and the difference between GL and SL distributions becomes negligible.

\begin{figure}[h!]
    \centering
    \includegraphics[width=1\linewidth]{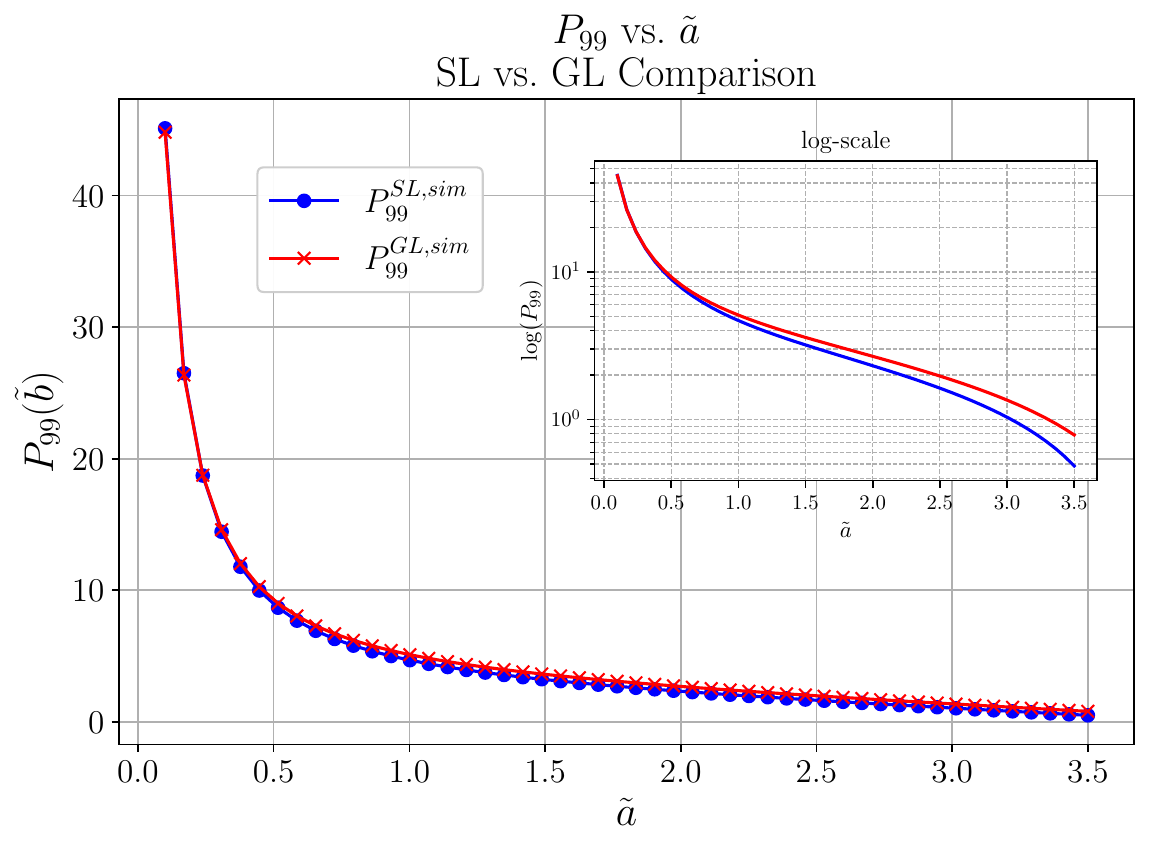}
    \caption{GL influence in the BESS PDF for different $\tilde{a}$ values with a fixed weight ratio parameter $c=0.25$. $P_{99}$ curve comparison for SL (blue) and GL (red).}
    \label{fig: SL vs GL p99}
\end{figure}

Figure~\ref{fig: RQ3.3} shows how the analytical solution compares with three cases of GL-distributed primary power increments $Y$, ranging from a very small $c$-value ($c=0.05$) to the case ($c=0.5$) where half of all events correspond to very small changes and, correspondingly, the long tail becomes much heavier. The dominant effect continues to be the increase of $\tilde{b}_{99}$ at large $\tilde{a}$-values. It can be seen that at, say, $\tilde{a}=3$ the $\tilde{b}_{99}$-value for the $c=0.5$ case is about twice as high as in the SL-case (represented by the analytic solution), in accordance with the findings of Probst \cite{Probst_2025}. However, given that the overall values are very small, this effect can be largely neglected. It can also be noted that $\tilde{b}_{99}$ in the $c=0.5$ case increases beyond the analytical prediction at very small values of $\tilde{a}$. However, such small values, as mentioned before, are completely irrelevant in practice since they would correspond to the requirement of a practical flat output of the VRE plant.

\begin{figure}[h!]
    \centering
    \includegraphics[width=0.9\linewidth]{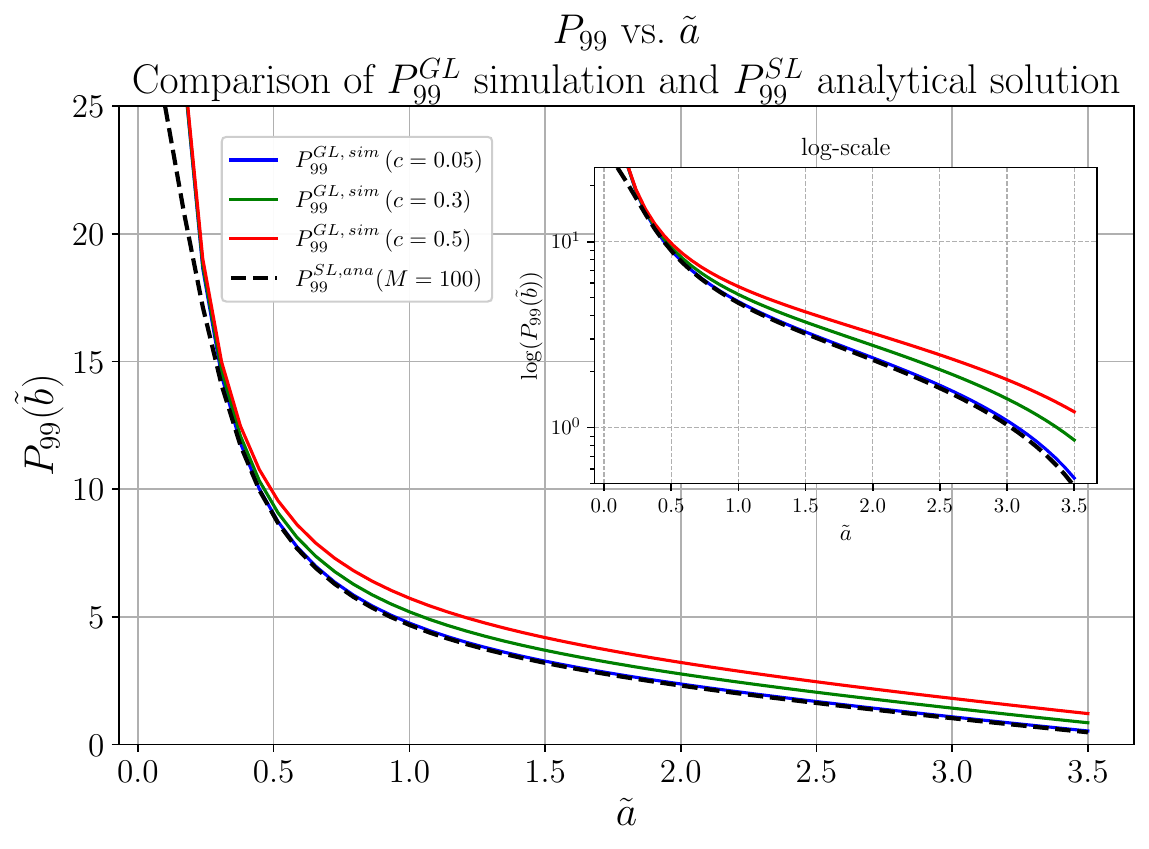}
    \caption{Impact of weight ratio parameter $c$ on GL-distribution power changes for BESS PDF (colored lines) showing the $P_{99}$ curve comparison against the SL-analytical Neumann series solution (black line).}
    \label{fig: RQ3.3}
\end{figure}

\subsection{Sizing of BESS inverters}

Having found the BESS power PDF from the solution to the integral equation using Neumann series, the required BESS power capacity can be determined easily in terms of a suitable P-value of the distribution. We first recall that the BESS power PDF is given by $g(\tilde{b},\tilde{a})=p_0(\tilde{a})\ u(\tilde{b},\tilde{a})$, where $u(\tilde{b})$ and $p_0(\tilde{b})$ are given by equations (\ref{eq:rescaled_BESS_PDF}) and (\ref{NormalizationCond}), respectively. One can then calculate the survival function, 
\begin{equation}
    S(\tilde{b},\tilde{a})=1-G(\tilde{b},\tilde{a})
\end{equation} 

where the BESS power CDF $G(\tilde{b})$ is given by
\begin{equation}
    G(\tilde{b};\tilde{a})=p_0+\int_0^{\tilde b} g(b';\tilde{a})\ \mathrm{d}b'.
\end{equation}

The result can be expressed in terms of the upper incomplete Gamma function, yielding
\begin{equation}
    S(\tilde b;\tilde a)=p_0 \sum _{n=0}^\infty e^{-(n+1)\tilde a}\sum_{k=0}^n \bar{\Lambda}_{n,k} \Gamma \left(k+1, \tilde b+ \tilde a\right).
\end{equation}

One may then define the normalized BESS power capacity $\tilde{b}_0$ by solving the percentile equation for $\tilde{b_q}$,
\begin{align}
    S(\tilde{b}_0,\tilde a)=S(\tilde{b}_q,\tilde a)=1-q,
\end{align}

where $q\in(p_0,1)$. Typically, $\tilde{b}_0 = \tilde{b}_{0.99}$, i.e., the battery power capacity is taken as the P99 value of the battery power PDF. We will now apply these formulas to the case of truncated analytical solutions, emulating the work of a power systems practitioner.

\subsection{Truncated analytical solutions: A practical sizing example}

As stated at the outset of this manuscript, the availability of a simple analytical formula goes a long way towards allowing the power systems practitioner to size a BES system for ramp-rate control without having to recur to black-box simulations. This will be illustrated with a simple example, emulating the work of a system operator tasked with conducting an initial BESS sizing study for a wind power plant with the characteristic parameters shown in Table \ref{tab: Sim params}. To do so, the operator refers to its sizing manual, which specifies that equations (\ref{eq:three_term_approximation_1}), (\ref{eq:three_term_approximation_2}), and (\ref{eq:three_term_approximation_3}) shown below should be used for the purposes of indicative studies and that a safety factor of 20\% should be applied after the calculations. Considering the case $M=2$, the following equations are obtained from (\ref{eq:rescaled_BESS_PDF}) and (\ref{NormalizationCond}):

\begin{equation}
u(\tilde{a},\tilde{b})=\frac{1}{16}e^{-3\tilde{a} - \tilde{b}}\ p(\tilde{a},\tilde{b}),
\label{eq:three_term_approximation_1}
\end{equation}
where
\begin{equation}
    p(\tilde{a},\tilde{b}) = \tilde{b}^2 + (2+4e^{\tilde{a}}+4\tilde{a})b + [1+3\tilde{a}^2 + 3\tilde{a} + e^{\tilde{a}}(2+8e^{\tilde{a}} +4\tilde{a})]
    \label{eq:three_term_approximation_2}
\end{equation}

The point mass is given by:
\begin{equation}\label{eq:three_term_approximation_3}
p_0(\tilde{a})=
\frac{16 e^{3\tilde{a}}}{
    5
    + 7\tilde{a}
    + 3\tilde{a}^{2}
    + 6 e^{\tilde{a}}
    + 4\tilde{a} e^{\tilde{a}}
    + 8 e^{2\tilde{a}}
    + 16 e^{3\tilde{a}}
}.
\end{equation}

From previous impact studies, the system operator has concluded that a ramp rate limit of $r=1\%/\mathrm{min}=60\%/\mathrm{h}$ would be tolerable for new wind farms. This translates into a tolerable change in grid power per 1-min time step ($\Delta t=1/60$ h) of $a=r P_\mathrm{max}=1.5 \ \mathrm{MW/min} = 90$ MW/h. From the wind resource information submitted for the wind project developer, the system operator concludes that the characteristic scale of an equivalent Laplace distribution is given by $\beta=0.6$ $\mathrm{MW}^{-1}$, leading to the normalized tolerable slope of $\tilde{a}=0.9$ as before. It should be recalled that these expressions are universally valid, i.e., they do not depend on the characteristics of the variability at the wind or solar plant site (to the extent that higher-order effects such as a non-vanishing autocorrelation and heteroskedasticity can be neglected), since both the battery power and the critical slope have been expressed in terms of normalized variables ($\tilde b$ and $\tilde a$) only. 

Substituting the parameter values and de-normalizing the BESS power yields the following three-term ($M'=M+1 = 3$) approximation to the BESS power distribution:
\begin{equation}
    \hat{g}(b)=(0.1266+ 0.01716 b +0.00067 b^2)\ e^{-0.6\ b}.
\end{equation}

Similarly, the CDF takes the form:
\begin{align}
    G(b)=1-e^{-\beta b} \left( c_0+c_1 b+c_2 b^2\right), \\
    (c_0,c_1,c_2)=(0.2647, \ 0.0323, \ 0.0011).
    \label{eq:3-term CDF}
\end{align}

From Eq.~(\ref{eq:3-term CDF}) the percentiles shown in the first data column of Table \ref{tab:bq_power} can be calculated. It is conspicuous from the table that the 3-term approximation somewhat underestimates all $p$-values compared to the solution with $M=100$. Compared to the results of the algorithmic solution with a finite wind farm limit ($P_\mathrm{max}=150$ MW), corresponding to a wind farm capacity in natural units of $p_\mathrm{max}=\beta P_\mathrm{max}=90$, this discrepancy diminishes. If also the safety factor of $20\%$ is applied to the results, following common engineering practice, then the predicted results become similar to the results of the algorithmic simulation. It can therefore be concluded that even a 3-term approximation of the analytical solution of the BESS sizing problem provides a reasonably good first estimate for the required power capacity of a battery energy storage system for ramp rate control.

\begin{table}[h!]
\centering
\caption{$p$-values for the analytical solution with $M=2$, with \textbf{}$M=100$, and comparison with the algorithmic solution with finite capacity}
\begin{tabular}{c c c c c}
\hline
 & \multicolumn{4}{c}{Required BESS capacity [MW]} \\
\cline{2-5}
\makecell{Percentile} & \makecell{Analytical \\ ($M=2$)} & \makecell{$M=2$ w/ \\safety factor} & \makecell{Analytical \\ ($M=100$)} & \makecell{Algorithmic \\ ($P_\mathrm{max}=150$\ MW)} \\
\hline
0.90 & 2.01 & 2.41 & 2.91 & 2.79\\
0.95 & 3.42 & 4.10 & 4.62 & 4.46\\
0.99 & 6.61 & 7.93 & 8.60 & 8.31\\
\hline
\end{tabular}
\label{tab:bq_power}
\end{table}

\section{Summary and conclusions}\label{sec:summary_conclusions}

In the present work, the stationary probability density function (PDF) of a reflected random work has been derived in an analytical way. The PDF was obtained from an integral equation that was derived within the framework of modern probability theory, based on the concepts of measure theory. The results can be viewed as a generalization of the Lindley equation for a single queue with a single server. The reflected random walk can be interpreted as a battery power process driven by a ramp control requirement for a Variable Renewable Energy (VRE) plant (such as wind or solar PV). In this work, the focus was laid on the control of negative ramp rates (usually the ones of interest to power system operators), but the framework derived is valid for both positive and negative ramp control. The analytical PDF obtained was derived under the assumption of Laplace-distributed variations of the primary power source (wind or solar), which has been shown in literature to be a good first approximation to short-term variability, considering time intervals in the range of one to a few minutes and independent and identically distributed power increments. The derivation is based on an infinite Neumann series, with the normalized critical slope $\tilde{a}=a\beta$ being the only free parameter. Here, $a$ [MW] is the tolerable change in grid power per time interval and $\beta^{-1}$ [MW] is the characteristic scale of the distribution of the primary power fluctuations. The analytical solution was shown to be practically indistinguishable from both a numerical solution of the integral equation and an algorithmic solution of the reflected random walk for $M=100$ terms of the Neumann series. However, a high accuracy can be obtained even with a relatively small ($\sim 15-20$) number of terms. In those cases, the evaluation of the analytical solution is significantly faster than the numerical solution. More importantly, a few terms are sufficient to estimate the required battery power capacity, defined as the P99 value of the battery power PDF, within ten percent or less. This opens the way for designing a simple methodology for the power systems practitioner, e.g., power systems operators or regulatory agencies, for the sizing of battery storage systems for ramp control. Such formulas, possibly enhanced by suitable safety factors, could be used in sizing manuals and regulatory documents, replacing the current black-box processes with a science-based approach.

To the best knowledge of the authors, the present work is the first analytical treatment of the sizing of battery energy storage systems for ramp rate control considering the correlation between consecutive battery power values and one of the few attempts to derive general sizing criteria. The current work is limited to the calculation of the required BESS power capacity, typically identified with the inverter capacity. For the sake of brevity, a similar treatment for the required energy capacity was not included at this time. The corresponding treatment will be described in a companion paper. As far as departures from the ideal nature of the distribution of the power increments are concerned, such as a finite degree of autocorrelation and the presence of heteroskedasticity, these effects will be investigated in future work.  

\appendix

\section{Technical lemmas and extended proofs}\label{app:proofs}

\noindent The proof of Lemma \ref{LEMMA UPDATE}:

\begin{proof}
    Recalling the definition of the BESS power law, for $A\in \mathcal{B}(S)$ 
    \begin{equation}
        \Pi_{n+1}(A)=\mathbb P (B_{n+1}\in A).
    \end{equation}
    
    \noindent We may express any measurable set as a random variable using indicator functions:
    \begin{equation}
        \Pi_{n+1}(A)=\mathbb P (B_{n+1}\in A)=\mathbb E \left[\textbf{1}_{A} (B_{n+1})\right]
    \end{equation}
    
    \noindent By the Tower Property, we may write:
    \begin{equation}
        \Pi_{n+1}(A)=\mathbb E \left[\textbf{1}_{A}(B_{n+1})\right]=\mathbb E\left[\mathbb E \left[\textbf{1}_{A}(B_{n+1})|\mathcal{F}_n\right]\right]
        \label{tower}
    \end{equation}

    \noindent The expectation of the indicator can be expressed as,
    \begin{equation}
        \mathbb E\left[\mathbb E \left[\textbf{1}_{A}(B_{n+1})|\mathcal{F}_n\right]\right]=\mathbb{E}[\mathbb{P}(B_{n+1}\in A\mid \mathcal{F}_n)]
    \end{equation}
        
    \noindent By the Markov Property,
    \begin{equation}
        \mathbb{P}(B_{n+1}\in A\mid \mathcal{F}_n)=\mathbb{P}(B_{n+1}\in A\mid B_n)=\kappa(B_n,A)\quad \text{a.s.}
    \end{equation}
    
    \noindent Hence,     
    \begin{equation}
        \Pi_{n+1}(A)=\mathbb E\left[\kappa(B_n,A)\right]
    \end{equation}
    
    \noindent Since $s\rightarrow \kappa(s,A)$ is measurable and bounded, and $B_n$ has law $\Pi_n$:
    \begin{equation}
        \mathbb E\left[\kappa(B_n,A)\right]=\int_S \kappa(s,A)\ \Pi_n(\mathrm{d}s).
    \end{equation}
    
    \noindent If $S\subseteq \mathbb R$, and  $A=(-\infty, b]$, then
    \begin{equation}
        G_{n+1}(b)=\Pi_{n+1}((-\infty,b])=\int_S \kappa(s,(-\infty, b])\ \Pi_n(\mathrm{d}s),
    \end{equation}
    
    \noindent so that the CDF evolution equation is given by
    \begin{equation}
        G_{n+1}(b)=\int_S F_\kappa (b|s)\ \Pi_n(\mathrm{d}s).
    \end{equation}
\end{proof}

\noindent The proof of remark \ref{REMARKSPLIT}:

\begin{proof}
    Let $S\subseteq \mathbb R$ be a measurable state space with Borel sigma-algebra $\mathcal{B}(S)$. We assume that the probability measures $\Pi_n$ and $\kappa(s,\cdot)$ have no singular continuous part and that the atoms lie in a fixed countable set, $C\subset S$. We had defined the reference measure:
    
    \begin{equation}
        \rho:=\lambda_{S\backslash C}+ \sum _{c\in C} \delta_c. 
    \end{equation}
    
    The Lebesgue Decomposition Theorem on $\mathbb R$ states that, for any finite measure $\mu$ on $S\subseteq \mathbb R$, there exist unique measures $\mu_{a.c.}\ll \lambda$ and $\mu_s \perp \lambda$ such that $\mu=\mu_{a.c.}+\mu_s$. Furthermore, $\mu_s$ splits uniquely into an atomic part, $\mu_{a.t.}$ and a singular-continuous part, $\mu_{s.c.}$. We assume that $\mu_{s.c.}=0$ for all relevant measures $\mu\in \{\Pi_n, \kappa(s,\cdot)\}$. Hence, for each $\mu$,
    \begin{equation}
        \mu =\mu_{a.c.}+\mu_{a.t.}, \quad \mu_{a.c.}\ll \lambda, \ \mu_{a.t.}=\sum_{c\in C}\mu(\{c\})\ \delta_c. 
    \end{equation}
    
    \noindent The decomposition then becomes
    \begin{equation}
        \mu(A)=\int_{A\cap (S\backslash C)}f_\mu \ \mathrm{d}\lambda+ \sum_{c\in C}\mu(\{c\})\ \delta_c(A),
    \end{equation}
    
    \noindent where $f_\mu$ is given by the RN derivative:
    \begin{equation}
        f_\mu =\frac{\mathrm{d}\mu_{a.c.}}{\mathrm{d}\lambda}.
    \end{equation}

    \noindent Given any $A\in \mathcal B(S)$ with $\rho(A)=0$. Then, by the definition of $\rho(A)=\lambda(A\cap (S\backslash C))+\sum_{c\in C} \delta_c(A)=0$, 
    \begin{align}
        \lambda(A\cap (S\backslash C))&=0, \\
        \delta_c(A)&=0 \text{ for every } c\in C.
    \end{align}
    
    \noindent Introducing such a set into the general decomposed measure $\mu$:
    \begin{equation}
        \mu(A)=\int_{A\cap (S\backslash C)} f_\mu \mathrm{d}\lambda+ \sum _{c\in C}\mu(\{c\})\ \delta_c(A)=0
    \end{equation}
    
    \noindent since $\lambda(A\cap (S\backslash C))=0$ and $\delta_c(A)=0 \ \forall c$. Since this holds for every $A$ with $\rho(A)=0$, then $\mu \ll \rho$. Since $\mu\in \{\Pi_n, \kappa(s, \cdot)\}$ was arbitrary, $\Pi_n\ll \rho$ and $\kappa(s,\cdot) \ll \rho$,
\end{proof}

\section{Coefficient recursion}\label{app:coeff}

Let $a=\alpha \Delta t$, so that $\tilde{a}:=\beta a= \beta \alpha \Delta t$ and $\tilde{w}:=\tilde{b}+\tilde{a}$; recalling the action of the integral operator as:

\begin{equation}
    (\mathbf{\tilde{K}}g)(\tilde b)=\frac{1}{2}\left[e^{-(\tilde b+\tilde a)}\int_0^{\tilde w}e^{\tilde s}\,\tilde g(\tilde s)\,\mathrm{d}\tilde s+e^{(\tilde b+\tilde a)}\int_{\tilde w}^\infty e^{-\tilde s}\,\tilde g(\tilde s)\,\mathrm{d}\tilde s\right],
\end{equation}

Further recalling the series representation of the repeated composition of $\mathbf{\tilde{K}}$:
\begin{equation}
    (\mathbf{\tilde{K}}^n \tilde{f_X})(\tilde b)=e^{-[\tilde b+(n+1)\tilde a]}\sum_{k=0}^n \bar{\Lambda}_{n,k}\,\tilde w^{\,k}.
\end{equation}
Recalling how the first term in the Neumann series was given by $\tilde{f_X}(\tilde{b})$:

\begin{equation}
    \tilde{u}_0=\frac{1}{2}e^{-[\tilde b+\tilde a]},
\end{equation}
so that $\bar{\Lambda}_{0,0}=\frac{1}{2}$. In order to find the recursive formula for $\bar{\Lambda}_{n,k}$, we define the basis functions:
\begin{align}
    \phi_{p,k}(\tilde s):=( \tilde s + \tilde a) ^{\,k} e^{-[\tilde s+p\tilde a]}, \ p\ge1, \ k\ge 0.
\end{align}

Then, applying the operator to the basis functions is equivalent to applying each "split" action separately:
\begin{equation}
\nonumber
    (\mathbf{\tilde{K}}\phi_{p,k})(\tilde b)=\frac{1}{2}\left[e^{-(\tilde b+\tilde a)}\int_0^{\tilde w}e^{\tilde s}\,(\tilde s + \tilde a )^{\,k} e^{-[\tilde s+p\tilde a]}\ \mathrm{d}\tilde s+e^{(\tilde b+\tilde a)}\int_{\tilde w}^\infty e^{-\tilde s}\,(\tilde s + \tilde a )^{\,k} e^{-[\tilde s+p\tilde a]}\ \mathrm{d}\tilde s\right],
\end{equation}

can be separated into:
\begin{equation}
    (\mathbf{\tilde{K}}^{(L)}\phi_{p,k})(\tilde b)=\frac{1}{2}e^{-(\tilde b+\tilde a)}\int_0^{\tilde w}e^{\tilde s}(\tilde s+\tilde a)^k e^{-[\tilde s+p\tilde a]}\ \mathrm{d}\tilde s=\frac{1}{2}e^{-(\tilde b+\tilde a)}L_k,
\end{equation}

\begin{equation}
    (\mathbf{\tilde{K}}^{(U)}\phi_{p,k})(\tilde b)=\frac{1}{2}e^{(\tilde b+\tilde a)}\int_{\tilde w}^\infty e^{-\tilde s}(\tilde s+\tilde a)^k e^{-[\tilde s+p\tilde a]}\ \mathrm{d}\tilde s=\frac{1}{2}e^{(\tilde b+\tilde a)}U_k.
\end{equation}

\noindent First, the integral $U_k$ will be explored in detail. 
\begin{align}
    U_k=e^{-p\tilde a}\int_{\tilde w}^\infty (\tilde s+\tilde a)^k e^{-2\tilde s}\ \mathrm{d}\tilde s=e^{-p\tilde a}I_k
\end{align}

\noindent Setting $u=\tilde{s}+\tilde{a}$, then $\tilde{s}=u-\tilde{a}$:
\begin{equation}
    I_k=\int_{\tilde w+\tilde a}^\infty u^k e^{-2(u-\tilde a)}\ \mathrm{d}u=e^{2\tilde a}\int_{\tilde w+\tilde a}^\infty u^k e^{-2u}\ \mathrm{d}u.
\end{equation}

\noindent We further have:
\begin{equation}
    I_k=e^{2\tilde a}\int_{L}^\infty u^k e^{-\lambda u}\ \mathrm{d}u=e^{2\tilde a}S_k,\quad L=\tilde w+\tilde a,\ \lambda=2.
\end{equation}

Finally, the integral $S_k$ can be solved using the incomplete gamma function, defined as $\Gamma(s,x)=\int_x^\infty t^{s-1} e^{-t}dt$. Defining $\zeta=\lambda u$, $d\zeta=\lambda du$:
\begin{equation}
    S_k=\int_{\lambda L}^\infty \left(\frac{\zeta}{\lambda}\right)^k \lambda e^{-\zeta}\ \mathrm{d}\zeta=\lambda^{-(k+1)}\int_{\lambda L}^\infty \zeta^k e^{-\zeta}\ \mathrm{d}\zeta. 
\end{equation}

\noindent Therefore, the solution to $S_k$ is simply:
\begin{equation}
    S_k=\lambda^{-(k+1)}\Gamma(k+1, \lambda L).
\end{equation}

\noindent Tracking back the value of $I_k$ and substituting $\lambda=2$, $L=\tilde{w}+\tilde{a}$:
\begin{equation}
    I_k=e^{2\tilde a}\left(\frac{1}{2}\right)^{k+1} \Gamma\!\left[k+1,\,2(\tilde w+\tilde a)\right].
\end{equation}

The upper incomplete gamma function satisfies the following series representation:
\begin{equation}
    \Gamma(s,x)=(s-1)!e^{-x} \sum_{j=0}^{s-1} \frac{x^j}{j!},
\end{equation}

\noindent if $s$ is a positive integer. In our case, $k\ge0$ is a positive integer, so that:
\begin{equation}
    \nonumber
    I_k=e^{2\tilde a}\left(\frac{1}{2}\right)^{k+1}k! \,e^{-2(\tilde w+\tilde a)} \sum_{m=0}^k \frac{(\tilde w+\tilde a)^m}{m!}\,2^{\,m}.
\end{equation}

\noindent Simplifying the equation above yields:
\begin{align}
    I_k=e^{-2\tilde w}\sum_{m=0}^k \frac{k!}{m!}(\tilde w+\tilde a)^m \left(\frac{1}{2}\right)^{k-m+1}.
\end{align}

\noindent Tracing back the value of $U_k$ leads to:
\begin{align}
\nonumber
    U_k=e^{-p\tilde a}e^{-2\tilde w}\sum_{m=0}^k \frac{k!}{m!}(\tilde w+\tilde a)^m \left(\frac{1}{2}\right)^{k-m+1},
\end{align}

\noindent where the complete solution to $(\mathbf{\tilde{K}}^{(U)}\phi_{p,k})(\tilde{b})$, with $\tilde{w}=\tilde{b}+\tilde{a}$ is given by:
\begin{align}
\nonumber
    (\mathbf{\tilde{K}}^{(U)}\phi_{p,k})(\tilde b)=\frac{1}{2}e^{-p\tilde a}e^{-(\tilde b+\tilde a)}\sum_{m=0}^k \frac{k!}{m!}(\tilde w+\tilde a)^m \left(\frac{1}{2}\right)^{k-m+1},
\end{align}

\noindent Simplifying the expression:
\begin{align}
\nonumber
    (\mathbf{\tilde{K}}^{(U)}\phi_{p,k})(\tilde b)=\frac{1}{2}e^{-(\tilde b+(p+1)\tilde a)}\sum_{m=0}^k \frac{k!}{m!}(\tilde w+\tilde a)^m \left(\frac{1}{2}\right)^{k-m+1}.
\end{align}

Recalling the Binomial Theorem, the polynomial $(\tilde{w}+\tilde{a})^m$ can be expanded as:
\begin{align}
    (\tilde w+\tilde a)^m=\sum_{r=0}^m\begin{pmatrix}
        m \\ r
    \end{pmatrix} \tilde a^{\,m-r}\,\tilde w^{\,r}.
\end{align}

\noindent Inserting this expansion into the previous expression results in:
\begin{align}
\nonumber
    (\mathbf{\tilde{K}}^{(U)}\phi_{p,k})(\tilde b)=\frac{1}{2}e^{-(\tilde b+(p+1)\tilde a)}\sum_{m=0}^k \frac{k!}{m!}\left[\sum_{r=0}^m\begin{pmatrix}
        m \\ r
    \end{pmatrix} \tilde a^{\,m-r}\,\tilde w^{\,r}\right] \left(\frac{1}{2}\right)^{k-m+1}.
\end{align}

\noindent Finally, the coefficients for the upper contribution can be found:
\begin{align}
\nonumber
    (\mathbf{\tilde{K}}^{(U)}\phi_{p,k})(\tilde b)=e^{-(\tilde b+(p+1)\tilde a)}\sum_{r=0}^k \left[\sum_{m=r}^k \frac{1}{2}\frac{k!}{m!} \begin{pmatrix}
        m \\ r
    \end{pmatrix} \tilde a^{\,m-r} \left(\frac{1}{2}\right)^{k-m+1}\right] \tilde w^{\,r}.
\end{align}

\noindent So that the coefficients are expressed by:
\begin{align}
\nonumber
   U_{k\rightarrow r}=\sum_{m=r}^k \frac{1}{2}\frac{k!}{m!} \begin{pmatrix}
        m \\ r
    \end{pmatrix} \tilde a^{\,m-r} \left(\frac{1}{2}\right)^{k-m+1},
\end{align}

where $r=0,1,\cdots, k$. Now returning to the lower coefficients' contribution:
\begin{equation}
    (\mathbf{\tilde{K}}^{(L)}\phi_{p,k})(\tilde b)=\frac{1}{2}e^{-(\tilde b+\tilde a)}\int_0^{\tilde w}e^{\tilde s}(\tilde s+\tilde a)^k e^{-[\tilde s+p\tilde a]}\ \mathrm{d}\tilde s=\frac{1}{2}e^{-(\tilde b+\tilde a)}L_k,
\end{equation}

\noindent where $L_k$ is given by:
\begin{equation}
    L_k =\int_0^{\tilde w}(\tilde s+\tilde a)^k e^{- p\tilde a}\ \mathrm{d}\tilde s=e^{- p\tilde a}\frac{([\tilde w+\tilde a]^{k+1}-\tilde a^{\,k+1})}{k+1},
\end{equation}

\noindent so that the complete lower piece solution is given by:
\begin{align}
    ({K}^{(L)}\phi_{p,k})(\tilde b)=\frac{1}{2} e^{-[\tilde b+(p+1)\tilde a]}\frac{([\tilde w+\tilde a]^{k+1}-\tilde a^{\,k+1})}{k+1}.
\end{align}

\noindent Applying the binomial expansion for $(\tilde{w}+\tilde{a})^{k+1}-\tilde{a}^{k+1}$, it is possible to express the aforementioned expression as:
\begin{align}
    (\mathbf{\tilde{K}}^{(L)}\phi_{p,k})(\tilde b)=e^{-[\tilde b+(p+1)\tilde a]}\sum_{r=1}^{k+1}\frac{1}{2}\frac{1}{k+1} \begin{pmatrix}
        k+1 \\ r
    \end{pmatrix} \tilde a^{\,k+1-r} \tilde w^{\,r},
\end{align}

\noindent so that the coefficient contribution for the lower piece is given by:
\begin{align}
    L_{k\rightarrow r}=\frac{1}{2}\frac{1}{k+1} \begin{pmatrix}
        k+1 \\ r
    \end{pmatrix} \tilde a^{\,k+1-r},
\end{align}

\noindent for $r=1,\cdots, k+1$ and where $L_{k\rightarrow 0}=0$.

\noindent Finally, the coefficients are found to be:
\begin{align}
\bar{\Lambda}_{n+1,r}=\sum_{k=0}^n \bar{\Lambda}_{n,k}(L_{k\rightarrow r} + U_{k\rightarrow r}).
\end{align}
\noindent where $r=0,1,\dots,n+1$ (with $L_{k\to 0}=0$ and $U_{k\to (k+1)}:=0$ by convention).

\section{Point mass series}\label{app:point}

\noindent Recalling the expression for the point mass:
\begin{equation}
    p_0 \;=\; \frac{1}{\,1+\int_{0}^{\infty}\sum_{n=0}^{\infty}
        e^{-(\tilde v+(n+1)\tilde a)}
        \sum_{k=0}^{n}\bar{\Lambda}_{n,k}\,(\tilde v+\tilde a)^{k}\,d\tilde v\,}
    \;=\; \frac{1}{1+\tilde\Omega}.
\end{equation}

\noindent It is possible to simplify the following integral expression
\begin{equation}
    \tilde\Omega \;=\; \int_{0}^{\infty}\sum_{n=0}^{\infty}
        e^{-(\tilde b+(n+1)\tilde a)}
        \sum_{k=0}^{n}\bar{\Lambda}_{n,k}\,(\tilde b+\tilde a)^{k}\,\mathrm{d}\tilde b.
\end{equation}

\noindent In this case, it is possible to swap the integral and summations due to absolute convergence in the $L^1$ space :
\begin{equation}
\nonumber
    \tilde\Omega \;=\; \sum_{n=0}^{\infty} e^{-(n+1)\tilde a}
        \sum_{k=0}^{n}\bar{\Lambda}_{n,k}
        \int_{0}^{\infty} e^{-\tilde b}\,(\tilde b+\tilde a)^{k}\, \mathrm{d}\tilde b
    \;=\; \sum_{n=0}^{\infty} e^{-(n+1)\tilde a}
        \sum_{k=0}^{n}\bar{\Lambda}_{n,k}\,\tilde H_k .
\end{equation}

\noindent Now representing the integral $H_k$ as a sum by use of the binomial theorem:
\begin{equation}
    (\tilde b+\tilde a)^{k}\;=\;\sum_{j=0}^{k}\binom{k}{j}\,\tilde a^{\,k-j}\,\tilde b^{\,j}.
\end{equation}

\noindent Substituting the expansion into the integral:
\begin{equation}
    \tilde H_k \;=\; \int_{0}^{\infty} e^{-\tilde b}
        \sum_{j=0}^{k}\binom{k}{j}\,\tilde a^{\,k-j}\,\tilde b^{\,j}\, \mathrm{d}\tilde b .
\end{equation}

\noindent Finally, taking the summation outside the integral:
\begin{equation}
    \tilde H_k \;=\; \sum_{j=0}^{k}\binom{k}{j}\,\tilde a^{\,k-j}
        \int_{0}^{\infty} e^{-\tilde b}\,\tilde b^{\,j}\, \mathrm{d}\tilde b
    \;=\; \sum_{j=0}^{k}\binom{k}{j}\,\tilde a^{\,k-j}\, j!.
\end{equation}

\noindent Therefore, the integral expression $\Omega$ can be expressed in terms of infinite series:
\begin{equation}
    \tilde\Omega \;=\; \sum_{n=0}^{\infty} e^{-(n+1)\tilde a}
        \sum_{k=0}^{n}\bar{\Lambda}_{n,k}
        \sum_{j=0}^{k}\binom{k}{j}\,\tilde a^{\,k-j}\, j!.
        \label{ANALYTICAL:OMEGA}
\end{equation}

\section{Numerical solution to the integral equation}\label{app: Nystrom}

A common numerical method used in the solution of integral equations is the \textbf{Nyström Method}. In this method, the continuous problem is broken into discrete intervals, and the integral is replaced with a weighted sum. The kernel follows the familiar Toeplitz structure. A discrete Toeplitz matrix is an $N\times N$ matrix that is defined as:
\[A_{ij}=c_{i-j},\]
for constants $c_{1-n}\cdots, c_{n-1}$ and total number of elements $N$. This provides the Toeplitz matrix its characteristic diagonal structure. The discrete convolution operation can be constructed as a matrix multiplication problem, where one of the inputs is converted into a Toeplitz matrix. Whenever there exists a bi-infinite Toeplitz matrix, where the index sets of the constants are infinite in nature (for example, the naturals), then the matrix induces a linear operator. The linear operator is known as a Toeplitz operator. 

In our case, the kernel only depends on the difference $b-s$. Because of this, $\textbf{K}$ is a Toeplitz operator on the half-line. This fact is clear when the problem is discretized for a numerical solution, approximating the operator in matrix form. Defining the discrete sequences $b_i=ih$ where $h=B_{max}/N$ and $i=0, \cdots, N$. Further computing trapezoid weights for quadrature:

\begin{align}
    w=\left(\frac{1}{2}, 1\cdots, 1, \frac{1}{2}\right).
\end{align}

This will turn every integral to a weighted sum such that:

\begin{align}
    \int_0^{B_{max}}\cdot \ \mathrm{d}s = h \sum _j w_j(\cdot)_j,
\end{align}

where $B_{max}$ is the upper limit that should numerically be representative of infinity. The integral operator, \textbf{K}, can be constructed by use of this trapezoid quadrature; the Toeplitz structure given by $b-s$ should be recovered by defining a difference grid. The difference grid can be calculated by performing subtractions between values in $b_i$ (the battery value discrete grid):
\begin{align}
    D_{ij}=(i-j)h,
\end{align}

where we recall that $i$ and $j$ range from $0$ to $N$. Then, the \textbf{K} operator can be represented by:
\begin{align}
    (\textbf{K}g)_i\approx(\textbf{K}^{N}g)_i=h\sum_{j=0}^Nf\left((i-j)h\right)\ w_j g_j,
\end{align}

where $f$ is the given form of the PDF, $f_X$. One may define the vector form of $f_X(b)$ by evaluating it using the discrete grid:
\begin{align}
    \vec{f}=f_X(b_i).
\end{align}

Then, the linear system defined earlier can be solved using a numerical resolvent:
\begin{align}
    (\textbf{I}-\textbf{K}^{N})\vec{u}=\vec{f}, \\
    \vec{u}\approx (\textbf{I}-\textbf{K}^{N})^{-1}\vec{f}.   
\end{align}

The value of $p_0$ should be determined using the same quadrature and the normalization condition. Discretely, the integral in Eq.~\eqref{NormalizationCond} can be represented using the quadrature:
\begin{align}
    \Omega^{(N)}=h\sum_jw_ju_j,
\end{align}

where $u_j$ are the elements of $\vec{u}$. Then, the point mass is given by:
\begin{align}
    p_0\approx\frac{1}{1+\Omega^{(N)}}.
\end{align}

\begin{align}
    \vec{u}^{(k+1)}=\vec{f}+\textbf{K}^N\vec{u}^{(k)},
\end{align}

where $k=1,\cdots,K$ and $K$ is the total number of iterations. Therefore, the Piccard method is implementing the approximation:
\begin{align}
    \vec{u}\approx \sum_{n=0}^K \vec{f}^{*(n+1)}=\vec{f}+\vec{f}*\vec{f}+\vec{f}*\vec{f}*\vec{f}+\cdots.
\end{align}

\newpage

\begingroup
\hypersetup{allcolors=black}
\bibliographystyle{apsrev4-2} 
\bibliography{references}
\endgroup

\end{document}